\pdfoutput=1

\documentclass[sigplan,authorversion]{acmart}

\copyrightyear{2021}
\acmYear{2021}
\setcopyright{acmcopyright}
\acmConference[PPoPP '21]{26th ACM SIGPLAN Symposium on Principles and Practice of Parallel Programming}{February 27-March 3, 2021}{Virtual Event, Republic of Korea}
\acmBooktitle{26th ACM SIGPLAN Symposium on Principles and Practice of Parallel Programming (PPoPP '21), February 27-March 3, 2021, Virtual Event, Republic of Korea}
\acmPrice{15.00}
\acmDOI{10.1145/3437801.3441601}
\acmISBN{978-1-4503-8294-6/21/02}

\startPage{1}

\usepackage{booktabs}   
\usepackage{subcaption} 
\usepackage{listings}
\usepackage{amsthm}
\usepackage{enumitem}
\usepackage{multirow}
\usepackage{float}

\theoremstyle{definition}

\definecolor{codegreen}{rgb}{0,0.6,0}
\definecolor{codegray}{rgb}{0.5,0.5,0.5}
\definecolor{codepurple}{rgb}{0.58,0,0.82}

\lstdefinestyle{nanos_style}{
    basicstyle=\ttfamily\scriptsize,
    numberstyle=\tiny\color{codegray},
    breakatwhitespace=false,
    breaklines=true,
    captionpos=b,
    keepspaces=true,
    numbers=left,
    numbersep=2pt,
    showspaces=false,
    showstringspaces=false,
    showtabs=false,
    tabsize=2,
	language=C,
	float
}

\begin{document}

\title[Advanced Synchronization Techniques for Task-based Runtime Systems]{Advanced Synchronization Techniques for Task-based Runtime Systems}         


\author{David Álvarez}
\orcid{0000-0002-4607-1627}             
\affiliation{
  \institution{Barcelona Supercomputing Center}            
  \city{Barcelona}
  \country{Spain}                    
}
\email{david.alvarez@bsc.es}          

\author{Kevin Sala}
\orcid{0000-0001-8233-1185}             
\affiliation{
  \institution{Barcelona Supercomputing Center}            
  \city{Barcelona}
  \country{Spain}                    
}
\email{kevin.sala@bsc.es}         

\author{Marcos Maroñas}
\orcid{0000-0003-1035-5625}             
\affiliation{
  \institution{Barcelona Supercomputing Center}            
  \city{Barcelona}
  \country{Spain}                    
}
\email{marcos.maronasbravo@bsc.es}         

\author{Aleix Roca}
\orcid{0000-0002-6715-3605}             
\affiliation{
  \institution{Barcelona Supercomputing Center}            
  \city{Barcelona}
  \country{Spain}                    
}
\email{arocanon@bsc.es}         

\author{Vicenç Beltran}
\orcid{0000-0002-3580-9630}             
\affiliation{
  \institution{Barcelona Supercomputing Center}            
  \city{Barcelona}
  \country{Spain}                    
}
\email{vbeltran@bsc.es}         

\begin{abstract}
	Task-based programming models like OmpSs-2 and OpenMP provide a flexible data-flow execution model to exploit dynamic, irregular and nested parallelism.
	Providing an efficient implementation that scales well with small granularity tasks remains a challenge, and bottlenecks can manifest in several runtime components.
	In this paper, we analyze the limiting factors in the scalability of a task-based runtime system and propose individual solutions for each of the challenges, including a wait-free dependency system and a novel scalable scheduler design based on delegation.
	We evaluate how the optimizations impact the overall performance of the runtime, both individually and in combination.
	We also compare the resulting runtime against state of the art OpenMP implementations, showing equivalent or better performance, especially for fine-grained tasks.
\end{abstract}



\maketitle

\section{Introduction}
\label{sec:introduction}
Due to diminishing returns on modern CPUs' single-thread performance, the industry has shifted towards many-core and heterogeneous architectures \cite{DarkSiliconMulticore,TheisMoore,BorkarFuture}.
The recent focus on energy efficiency has increased the interest in systems with numerous processing elements with lower frequencies.
Those parallel systems can achieve huge performance figures, but their limiting factor is the scalability of the software. 

One of the most widely used standards for programming shared-memory systems in both industry and academy is OpenMP\cite{openmpindustry}.
OpenMP initially had only a \textit{fork-join} execution model, where programmers explicitly define parallel regions.
The fork-join model is an efficient way to exploit well-structured parallelism, but it is not well suited to exploit irregular, dynamic, or nested parallelism.
In recent years, task-based parallelism has been introduced in OpenMP to overcome these limitations.

The task-based paradigm can exploit more fine-grained, dynamic, and irregular parallelism than the fork-join model.
Additionally, it minimizes the need for global synchronization points, and it naturally copes with load-imbalance.
These features make the paradigm especially promising to exploit modern many-core architectures \cite{taskBased1, taskBased2}.

Moreover, the introduction of task data dependencies was the critical element to truly move forward to a \textit{data-flow} execution model that relies on fine-grained synchronizations between tasks.
This model reduces further the need for global synchronization points and allows the runtime to exploit data-locality between tasks.
However, task management inside the runtime might incur some non-negligible overhead, especially for fine-grained tasks on large many-core systems.

This paper presents optimized designs for the main components of a task-based runtime system.
We also present a lightweight and integrated instrumentation system, which we used to analyze in detail how the scalability of each component affects the global scalability of the runtime.

A task-based runtime system has three main components: the dependency system, the task scheduler, and the memory allocator, which tightly interact with each other.
The first stage of a task's life cycle is its creation, which involves the memory allocator.
The runtime then checks its data dependencies to determine if the task is ready or blocked based on the previous tasks' dependencies.
Once all its dependencies are satisfied, the task becomes ready and is added to the scheduler, which will eventually schedule it on an available core.
Once the task has executed, it releases its dependencies so that its successor tasks may become ready.
Note that the three components require a synchronization mechanism as they have to deal with multiple requests simultaneously.
Thus, the application developer has to strike a balance in task granularity: it has to be small enough to provide sufficient work for all available cores while being coarse enough to evade runtime system overheads \cite{OpenMPGranularities}.
However, as applications scale out to more cores (or nodes) and the problem size remains constant, task granularities naturally decrease.
At some point in the scaling process, tasks can become so small that the application is overhead-bound, and the scalability depends on the ability of the parallel runtime to handle small tasks.

Our contributions are to (1) present a novel wait-free data structure and algorithm to support complex dependency models in a task-based runtime; (2) provide a scalable task scheduler that works well under high contention and uses delegation instead of work-stealing; and (3) analyze and create a detailed performance profile of the task-based runtime with a lightweight instrumentation framework.
\section{Data dependency system}
\label{sec:dependencysystem}

The data dependency system is one of the limiting factors in the scalability of task-based runtimes, especially when running programs with very fine-grained tasks.
Such programs register large amounts of small tasks with dependencies that take a short time to execute.
Thus the overhead in the dependency system can significantly impact the overall application performance.

In this paper, we apply our optimizations to the Nanos6 runtime \cite{bsc2019nanos6} for the OmpSs-2 programming model.
Compared to OpenMP, the model for data dependencies in OmpSs-2 is more complex \cite{JMWeaks, reductions-ferran, bsc2020ompss2}.
The main complexity increase is because the dependency domains of tasks on different nesting levels can share dependencies,
 which complicates the locking scheme used to implement the model.
Reductions also are treated as data dependencies on OmpSs-2 tasks, unlike OpenMP where they are defined at a task group level.

\subsection{Dependencies in Nanos6}

\begin{figure}
	\centering
	\includegraphics[width=.8\columnwidth]{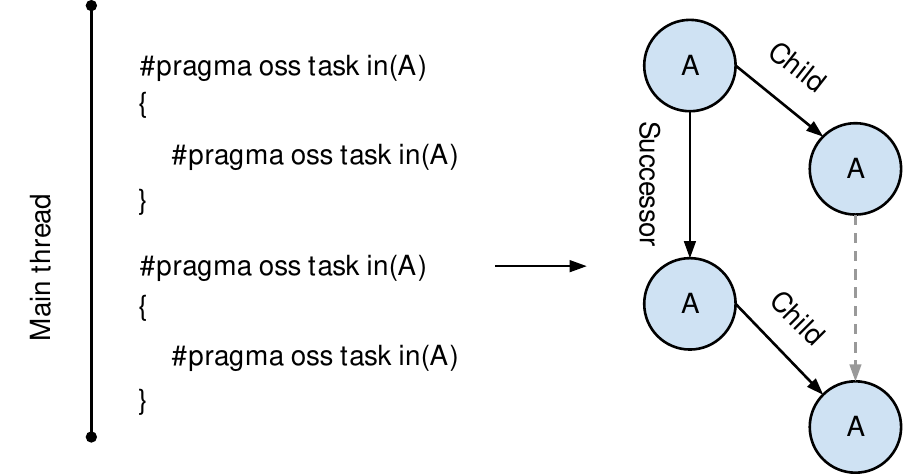}
	\Description{The graph of task dependency accesses on an OmpSs-2 program}
	\caption{The graph of task dependency accesses on an OmpSs-2 program. The program (left) results in the dependency graph (right) between the accesses to location A.}
	\label{fig:dependencygraph}
\end{figure}

In Nanos6, a task is a \textit{sibling} of another when they are at the same nesting level.
Tasks declared inside another one (nested) are considered \textit{child} tasks.

The dependencies of a task in Nanos6 are represented as accesses, which are composed of a memory address and an access type, e.g., read or write.
Two accesses have a {\it successor} relation when they share the same memory address and their tasks have a {\it sibling} link.
Similarly, two accesses have a {\it child} relation if they share the same memory address and their tasks have a {\it child} link.
Task access relations on OmpSs-2 programs form binary trees between the linked tasks, as shown in Figure \ref{fig:dependencygraph}.
Note that on OpenMP users cannot express dependencies crossing nesting levels, and the \textit{child} relationship is not considered to determine the dependencies between tasks.

During an OmpSs-2 program, several tasks can be created and finished concurrently.
They have to propagate dependency information through the data structures to determine if the added tasks are ready and if the recently finished tasks allow any successor to become ready.
\subsection{Wait-free data dependencies}

The previous implementation of dependencies inside Nanos6 was based on fine-grained locking, but it was very complex to avoid possible deadlocks.
Instead of protecting the data structures that hold the information about the dependencies through mutual exclusion, our alternative is to adapt some wait-free programming concepts to create a data structure capable of supporting the concurrency we need.
Otherwise, if we had decided to build a data structure based on mutual exclusion, we would have to compromise either with the complexity of fine-grained locking or the performance degradation of coarse-grained locking.
The main goal is not to have wait-freedom as a requirement but to provide fast and scalable dependency registration.

In this section, we describe the concept behind the dependency implementation we propose for the Nanos6 runtime. 
In the following section, we will formalize the approach and prove its wait-freedom property. 

When a program creates a task, all the dependencies are registered inside the Nanos6 runtime using the \textit{DataAccess} structure, shown in Listing \ref{lst:dataAccess}.
The \textit{Task} structure stores all task-related information, including a pointer to the array of its accesses.
Each access has an atomic \textit{flags} field that stores its current state, indicating if the dependency is currently satisfied (not preventing the task execution) and whether the satisfiability information has propagated to its successor and child accesses.
The access also stores a pointer to its \textit{successor}, which is the next access (belonging to a successor task) to the same address in the current nesting level, and a pointer to the \textit{child}, which points to the first access to the same address that belongs to a child task.

The flags field represents a Finite State Machine in which the state diagram has no cycles, so there are starting and final states.
Since we only modify this data structure with atomic operations, we have named it \textit{Atomic State Machine} or ASM.
Note that there is one instance of this state machine for each access.

\begin{lstlisting}[language=C, label=lst:dataAccess, caption=Relevant fields in the Task and DataAccess structures]
	struct Task {
		...
		DataAccess *dataAccesses;
	};

	struct DataAccess {
		void *address;
		std::atomic<access_flags_t> flags;
		DataAccessType type;
		DataAccess *successor, *child;
	};
\end{lstlisting}

The only way for an ASM to transition from one state to another is through receiving a data access message.
The structure of a message, shown in Listing \ref{lst:dataAccessMessage}, contains two flags fields.
One field contains the flags to set in the target access.
The other has flags that have to be set on the message's originator as a delivery notification.
The ASM's transitions have to be done as a single atomic operation, optimally through a \textit{fetch\&or}.
Based on the values before and after the transition, the ASM may generate additional messages to deliver to its child or successor accesses.
All the messages that are still pending to deliver are stored in a simple per-thread queue called MailBox.
We illustrate this process in Figure \ref{fig:mailbox}.

\begin{lstlisting}[language=C, label=lst:dataAccessMessage, caption=DataAccessMessage structure]
	struct DataAccessMessage {
		access_flags_t flagsForNext, flagsAfterPropagation;
		DataAccess *from, *to;
	};
\end{lstlisting}

Figure \ref{fig:mailbox} represents the basic structure of the algorithm.
While the \textit{MailBox} has undelivered messages, we pop one from the container and \textit{deliver} it to the destination access.
Upon receiving the message, the access atomically updates its \textit{flags} field.
The atomic update provides us with the flags' exact values before and after the message was received.
As flags cannot be unset, we know each transition happens only once in the access lifetime.
With this information, we decide if it is needed to generate more messages (to propagate information about satisfied accesses, for example).
Finally, we atomically update the \textit{flags} field of the originator access of the message to notify the delivery of the information.
We use this last atomic update to determine we can safely delete an access.

\begin{figure}
	\centering
	\includegraphics[width=1\columnwidth]{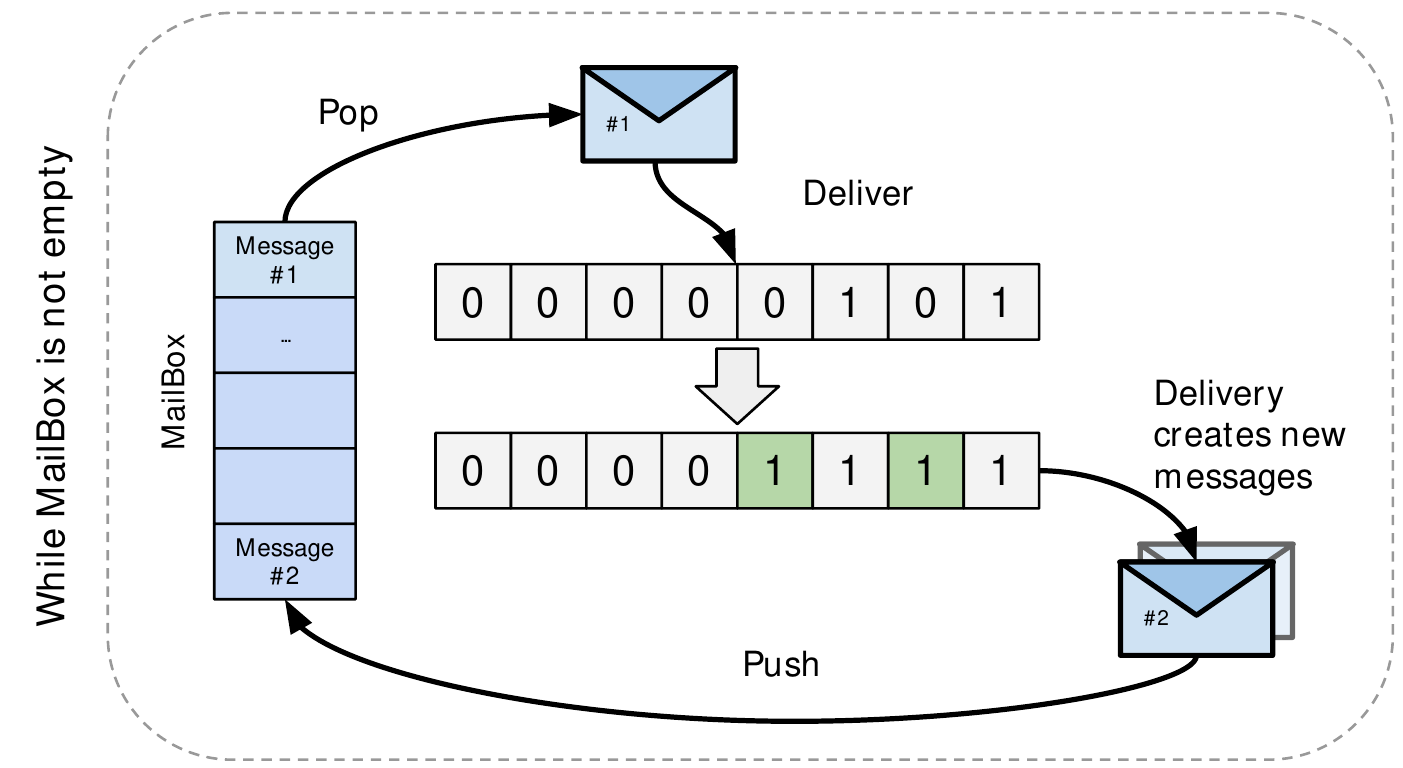}
	\Description{Atomic state machine dependency propagation}
	\caption{Atomic state machine dependency propagation}
	\label{fig:mailbox}
\end{figure}

\subsection{Formalization}

In this section, we introduce several relevant definitions and finally prove wait-freedom.

\theoremstyle{definition}
\begin{definition}{\textit{Access Flags.}}
We can define the set of all possible flags an access can have as set $F$.
Then, the set $F_a$ of flags that an access $a$ has is defined as $F_a \subseteq F$.
When $a$ is created, flags are initialized as $F_a = \emptyset$.
\end{definition}

\theoremstyle{definition}
\begin{definition}{\textit{Delivery.}}
	The only operation that can be done on the flags $F_a$ of an access $a$ is the delivery of a message $M \subseteq F$.
	A delivery operation is defined as:
		\[F_{a, i+1} = M \cup F_{a, i}\]
	Assuming the message M follows the two restrictions:
		\[ M \cap F_{a, i} = \emptyset \]
		\[ M \neq \emptyset \]
\end{definition}

The restrictions on the content of the messages are part of what provides the wait-freedom assurance.
Informally, bits in the flag field can only be set, and the field size is limited.
Additionally, each message that an access receives has to contain at least one flag, and none of the flags can be already set in the access.
By those properties, we can deduce that each access can receive only a limited number of messages, which in the worst case is $|F|$.

\begin{lemma}
	\label{lem:waitfreedelivery}
	The delivery of a message $M$ to an access $a$ is non-blocking and wait-free.
\end{lemma}

\begin{proof}
	To establish wait-freedom, we need to prove that there is a bound on the time a delivery operation can take regardless of any other threads.
	Assuming that the operation is performed using a \textit{CAS} primitive, we can assume constant time for the \textit{CAS}, but it can fail in case of conflict with another thread.
	Hence, to prove wait-freedom, we need to bound the number of failures due to conflict a delivery operation can suffer.

	A \textit{CAS} can fail if and only if another thread modified the memory location during the delivery operation.
	However, in our system, the only way to change the flags $F_a$ of an access $a$ is to deliver another message.
	As we established our restriction that a message cannot be empty, the maximum number of messages an access $a$ can receive if each message only contained one flag would be:
		\[M_{max} = \{ \{x\} \mid x \in F \}\]

	Then, the maximum number of \textit{CAS} operations that have to be done until one succeeds, assuming a worst case scenario, is $|M_{max}|$, which trivially $|M_{max}| = |F|$ .
	We can establish that the maximum number of tries to deliver a message $M$ is $T_m \leq |F|$.
	It is possible to get a closer bound on the retries if we consider the number of flags in the message to be delivered, but this is sufficient for us to prove Lemma \ref{lem:waitfreedelivery}.
	The time needed to deliver a message is clearly bounded to a constant number of \textit{CAS} operations.
\end{proof}

\theoremstyle{definition}
\begin{definition}{\textit{Unregister.}}
For a \textit{Task} $t$ that has a set of accesses $A_t$, the unregister operation on a Task is defined as delivering a specific message $M$ to each access $a$ so that $a \in A_t$.
\end{definition}

A task is unregistered once it finishes its execution, and the message delivered to the access indicates this condition.
An unregister operation will thus do $|A_t|$ delivery operations.
As we have proved that a delivery is wait-free, the unregister operation will be wait-free because it does a finite and known number of delivery operations.

\section{Task Scheduling System}
The scheduling system orchestrates the execution of ready tasks on worker threads.
Throughout this section, we assume that exactly one worker thread is bound to each CPU for simplicity but without loss of generality.
When a task becomes ready, it is forwarded to the scheduling system. Then, when a core becomes idle, it calls the scheduler to ask for more work.
If there are ready tasks, the scheduling system will determine the best task that can be executed on this specific core.
It is worth noting that multiple ready tasks can be added to the scheduler concurrently and that several worker threads can simultaneously call the scheduler.
Thus, it becomes mandatory to add some kind of synchronization on the scheduler to prevent data-races.

Most task-based runtime systems rely on multiple ready task queues combined with work-stealing to mitigate the above-mentioned problems.
However, on the typical application design pattern in which a single thread creates all tasks, work-stealing behaves similarly to the global lock approach because most threads need to steal work from a single creator queue.
In contrast, the approach described in this section adapts the global lock concept to handle both single creator and multiple creator cases efficiently.

Usi a global lock is the most straightforward approach to synchronize the scheduler.
In this case, the lock is acquired to add ready tasks to the scheduler and to schedule tasks to worker threads.
When task granularity is coarse enough, this approach works well and keeps the scheduling system's design simple and the scheduling policies accurate.

However, when task granularity is fine-grained and the system has many cores, the core that is creating new tasks might not be fast enough to feed all other cores.
In this scenario, many worker threads will busy-wait on the global lock.
This has two adverse effects.
Firstly, it increases contention on the cache subsystem due to the additional cache coherence traffic.
Secondly, it prevents ready tasks from entering the scheduler fast enough because the task creator has to fight with all of the worker threads to get the global lock.

A well-known technique to mitigate lock con\-ten\-tion on the scheduler is to let the worker threads spin for a while, and if they do not get any ready task, block its thread using a mutex until a ready task becomes available.
We avoid using this approach because it adds extra work to the thread that is creating tasks. When ready tasks are added to the scheduler, it has to check if there are blocked threads and wake them up with an expensive system call.

\subsection{Optimizing task insertion}
\label{sec:spsc}
To avoid the stagnation of ready tasks in the scheduler, we have decoupled the actions of adding and scheduling ready tasks. When a task becomes ready, we do not directly add it to the scheduler but a bounded wait-free single-consumer single-producer (SPSC) queue working as a buffer.

The number of SPSC queues can be configured from a single one to one per core.
In the first case, we would need a lock to synchronize all task additions, while in the latter, no locking is needed at all.
We use the lock to synchronize between producers, but the synchronization between producer and consumers remains wait-free.
In our experiments, we use one SPSC queue and lock per NUMA node.

When a worker threads enters the scheduler, it first drains all SPSC queues and inserts the ready tasks into the global ready queue.
With this approach, we ensure that any contention generated by many worker threads calling the scheduler does not affect the performance of cores that are creating tasks.
We can implement this optimization because the actual addition of tasks can be safely delayed until a core becomes idle and calls the scheduler.
Notice that this delegation technique is compatible with any lock implementation.

\subsection{Scalable lock designs}
The scheduler system has to be extensible, and adding new scheduling policies should be easy.
We have discarded a wait-free or lock-free scheduler because of its complexity and difficulty to maintain, as each scheduling policy would require a new ad-hoc design and implementation.

Our scheduling system relies on a global lock to protect its internal data structures, making it easy to develop new scheduling policies.
Ticket Locks \cite{10.1145/359060.359076} are fair and provide strict FIFO ordering, but they have contention problems under high-load conditions, so they are not suitable for our centralized scheduler.
Partitioned Ticket Locks \cite{10.1145/1989493.1989543} (PTLocks) extend Ticket Locks with a padded array used to do busy-waiting by idle threads.
If the size of this array is equal to the number of CPUs, then each core will busy-wait in a different array slot, reducing the cache coherence traffic to the minimum.
We use PTLocks as a building block of our optimized lock design presented in the next section.

\begin{lstlisting}[language=C++, label=lst:PTL, caption=Implementation of a PartitionedTicketLock, firstline=6, lastline=26]
struct PTLock {
  // Can be a constructor parameter
  const static int Size = 64;
  std::atomic<uint64_t> _head = {Size};
  uint64_t _tail = {Size + 1};
  std::atomic<uint64_t> _waitq[Size] = {{Size}};

  uint64_t _getTicket() {
    return _head.fetch_add(1);
  }
  void _waitTurn(uint64_t ticket) {
    while (_waitq[ticket % Size] < ticket) { spin(); }
  }
  void lock() {
    _waitTurn(_getTicket());
  }
  void unlock() {
    uint64_t idx = _tail % Size;
    _waitq[idx] = _tail++;
  }
};
\end{lstlisting}

Listing \ref{lst:PTL} shows the implementation of a Partitioned Ticket Lock. For the sake of clarity, we have omitted the padding of the fields to prevent false sharing, and the memory order constraints of all atomic operations. The {\it \_waitq} (line 6) is an array of unsigned 64-bit integers used to implement a circular buffer representing an infinite virtual waiting queue. The {\it \_head} and {\it \_tail} fields (lines 4 and 5) are used to index the {\it \_waitq} array. The {\it \_head} represents the index of the latest slot in the virtual waiting queue and the {\it \_tail} is the index of the next slot that will be able to acquire the lock. When the lock is free and no thread is waiting to acquire it, {\it \_tail == \_head+1}.

We initialize the lock such that {\it \_waitq[\_head\%Size] == \_head}, guaranteeing that the first thread that arrives will be able to acquire it.
The {\it lock()} operation (line 14) consists of just two calls. The first one is {\it \_getTicket()} (line 8), which performs an atomic fetch and increment of the {\it \_head} field to obtain the last ticket (line 9). The second call is {\it \_waitTurn} (line 11) that receives the ticket as a parameter. The current thread busy-waits on the {\it \_waitq[ticket \% Size]} position until it matches (or exceeds) the ticket value. The {\it unlock} operation (line 17) calculates the next slot index that will be able to acquire the lock.  Then it increments {\it \_tail} and writes {\it \_tail-1} in the computed slot to release the lock.

PTLocks perform as well as more complex designs such as MCS~\cite{10.1145/103727.103729} or Ticket Locks Augmented with a Waiting Array (TWA)~\cite{TWA}, however it requires more memory space.


\subsection{Delegation Ticket Lock (DTLock)}
\label{sec:dtlock}
Another well-known technique to improve the performance of data structures that are not amenable to fine-grained locking is delegation \cite{10.1145/3132747.3132771}.
The main idea behind this method is that protected data structures are accessed only by one privileged thread, which executes all the operations on behalf of the other threads.
Delegation relies on a lightweight and optimized communication protocol between the privileged thread and the rest of the threads to be able to delegate operations and forward results back.
A drawback of this approach is that it requires a dedicated core for each independent set of data structures that has to be protected, making it unpractical in many situations.

There are variants of the delegation technique that use queues to delegate work to the threads currently inside the lock \cite{queuedelegationlock}.
These are better suited to use in centralized schedulers as they do not require dedicated cores for each lock.
In order to build our scalable centralized scheduler, we have developed a novel Delegation Ticket Lock (DTLock) that builds on state of the art delegation techniques and extends our implementation of the PTLock with support for fine-grained and dynamic delegation of operations.
The DTLock supports the standard {\it lock}, {\it unlock} and {\it trylock} operations, as well as {\it lockOrDelegate}, {\it empty}, {\it front}, {\it popFront} and {\it setItem}.
The {\it lockOrDelegate} operation either acquires the lock if it is free or delegates the operation to the current lock's owner.
However, it is the current owner that will decide if it executes or not the delegated operations.
Suppose the current owner releases the lock without performing a pending delegated operation.
In that case, the thread that originally delegated that operation will eventually acquire the lock and execute it by itself.
The {\it empty}, {\it front}, {\it popFront} and {\it setItem} operations can only be called by the thread that owns the lock and are used to manage the threads that are waiting outside the lock.

The main advantages of the DTLock are that it does not require a dedicated core and its additional operations can be freely combined with traditional {\it lock}, {\it unlock} and {\it trylock} operations.
Additionally, the DTLock allows for threads to remain inside the critical section of the lock executing delegated operations.
This can be leveraged to minimize operation latency when there is not enough work to keep all cores busy.

\begin{lstlisting}[language=C++, label=lst:DTL, caption=Implementation of a Delegation Ticket Lock, firstline=5, float]
template <typename T>
struct DTLock : public PTLock {
  std::atomic<uint64_t> _logq[Size] = {};
  struct { uint64_t ticket; T item; } _readyq[Size];

  bool lockOrDelegate(uint64_t const id, T &item) {
    uint64_t const ticket = _getTicket();
    _logq[ticket % Size] = ticket + id;
    _waitTurn(ticket);
    if (_readyq[id].ticket != ticket) {
      _tail++;
      return true;
    }
    item = _readyq[id].item;
    return false;
  }
  bool empty() const {
    return _logq[_tail % Size] < _tail;
  }
  uint64_t front() {
    return _logq[_tail % Size] - _tail;
  }
  void popFront() {
    unlock();
  }
  void setItem(uint64_t const id, T item) {
    _readyq[id].item   = item;
    _readyq[id].ticket = _tail;
  }
};
\end{lstlisting}

Listing \ref{lst:DTL} shows the implementation of a DTLock in C++, which inherits the PTLock's {\it \_head}, {\it \_tail} and {\it \_waitq[]} members, as well as, {\it lock}, {\it unlock} and {\it tryLock} operations.
The DTLock extends the PTLock with two additional arrays. The {\it \_logq} array (line 3) is used to register waiting threads while the {\it \_readyq} array (line 4) is used to store the result of delegated operations, i.e. ready tasks in our case.

The first parameter of the {\it lockOrDelegate} operation is a unique \textit{id} that identifies the thread that is calling this function. This \textit{id} should be in the range {\it 0..Size-1} as it is used to index the {\it \_readyq} array. Thus, we need to know in advance the maximum number of threads that can call the DTLock. If the {\it lockOrDelegate} operation is finally delegated, the second parameter is used to store the result.

The first thing done in {\it lockOrDelegate} is to obtain a ticket (line 7).
Then, the thread is registered on the {\it \_logq} (line 8) array with just one store operation that combines the ticket and calling thread's id.
The values written on the {\it \_logq} array cannot be overrun because it is guaranteed that there will be at most {\it Size} threads calling the {\it lockOrDelegate} operation, so that each thread will have their own position.
Once the thread has been registered, it just busy-waits (line 9) until it acquires the lock (lines 11-12), or the operation has been delegated and the result is stored in the {\it \&item} parameter (line 14).

The {\it empty} operation (line 17) returns {\it true} if there is no thread registered in the {\it \_logq} array and {\it false} otherwise. We check the first position of the {\it \_logq} array, and if the value is smaller than {\it \_tail}, we know there is no thread waiting yet. Otherwise we would see the value written in line 8. Notice that this operation is intrinsically racy but harmless.

If a call to {\it empty} returns {\it false}, then the owner of the lock can call {\it front} (line 20) to obtain the id of the thread that is waiting. To that end, we only need to do the inverse of the operation done on line 8, subtracting the {\it \_tail}'s value to obtain the thread id (line 21).

Then, the {\it setItem} operation assigns a result {\it T} to a registered thread using its id to index the {\it \_readyq} array. First, it sets the {\it item} field (line 27) and then the {\it ticket} to the {\it \_tail} value, marking the entry as valid. We use the {\it ticket} field in line 10 to determine if an operation was delegated or not.
Finally, the {\it popFront} operation wakes up the first thread busy-waiting on the {\it \_waitq} by executing the {\it unlock} operation.

\subsection{Synchronized Scheduler}
This section presents a synchronized scheduler that leverages the SPSC queues and the DTLock described in sections \ref{sec:spsc} and \ref{sec:dtlock}, respectively.

\begin{lstlisting}[language=C++, label=lst:SynchSched, caption=Implementation of the Synchronized Scheduler using a Delegation Ticket Lock, firstline=11]
struct SyncScheduler {
  DTLock<Task *> _lock;
  UnsyncScheduler _sched;
  PTLock _addQueueLock;
  boost::lockfree::spsc_queue<Task *> _addQueue = {100};

  void processReadyTasks() {
    _addQueue.consume_all(
        [&](Task *t){ _sched.addReadyTask(t); });
  }
  void addReadyTask(Task *task) {
    while (1) {
      _addQueueLock.lock();
      bool added = _addQueue.push(task);
      _addQueueLock.unlock();
      if (added) break;
      if (_lock.tryLock()) {
        processReadyTasks();
        _lock.unlock();
      }
    }
  }
  Task *getReadyTask(uint64_t const id) {
    Task *task;
    if (!_lock.lockOrDelegate(id, task))
      return task;
    processReadyTasks();
    while (!_lock.empty()) {
      uint64_t waitingId = _lock.front();
      task = _sched.getReadyTask(waitingId);
      if (task == nullptr) break;
      _lock.setItem(waitingId, task);
      _lock.popFront();
    }
    task = _sched.getReadyTask(id);
    _lock.unlock();
    return task;
  }
};
\end{lstlisting}

Listing \ref{lst:SynchSched} shows the implementation of a synchronized scheduler using a DTLock (line 2) and a SPSC wait-free queue (line 5) to synchronize the {\it getReadyTask} and {\it addReadyTask} operations, respectively. The SyncScheduler is a wrapper of the unsynchronized scheduler (line 3), which implements the actual scheduling policy.

The PTLock (line 4) protects the producer side of the SPSC queue in the {\it addReadyTask} function (line 14), while the DTLock protects the consumer side in the {\it processReadyTasks} function. In the {\it addReadyTask} function, if there is no free space on the SPSC queue, the thread will try to acquire the DTLock with a {\it tryLock} operation (line 17). If it succeeds, it will call the {\it processReadyTasks} function (line 18), which removes all tasks from the SPSC queue and inserts them to the unsynchronized scheduler (line 8).

We use the {\it lockOrDelegate} operation (line 25) of the DTLock to synchronize {\it getReadyTask} operations.
If the operation is delegated because another thread owns the DTLock, the calling thread will busy-wait until it gets a task (lines 25-26).
Otherwise, it will eventually acquire the lock and call the {\it processReadyTasks} (line 27) to add the tasks that are waiting on the SPSC queue into the unsynchronized scheduler.
Then, it will try to schedule a task for each of the threads that are waiting on the DTLock (lines 28-33) until there are no more waiting threads (line 28) or no ready tasks are left (line 31).
At that point, it will try to get a task for him (line 35) and then release the DTLock (line 36).
In our simplified design, the thread inside the scheduler leaves as soon as there are no more tasks to schedule.
However, it is easy to extend our design in a way that {\it processReadyTasks} is called when no tasks are left inside the scheduler.

\begin{figure}
    \centering
    \includegraphics[width=.7\columnwidth]{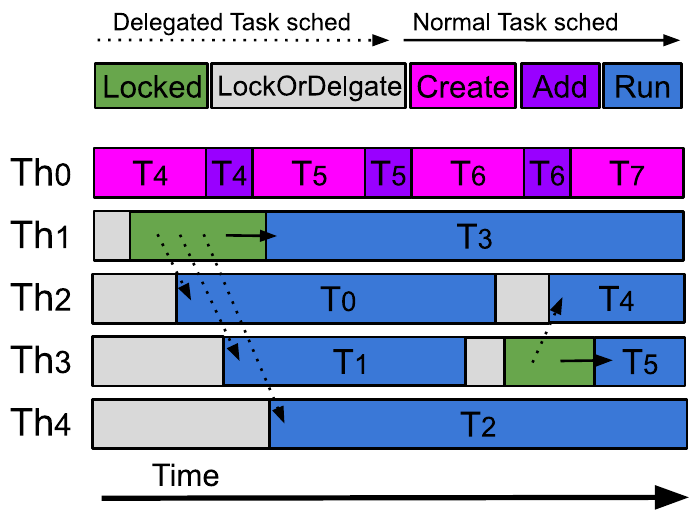}
	\caption{Timeline of five threads using a DelegationLock to add and get ready task into the scheduler.}
	\Description{Timeline of five threads using a DelegationLock to add and get ready task into the scheduler.}
    \label{fig:timelinesyncscheduler}
\end{figure}


Figure \ref{fig:timelinesyncscheduler} shows a timeline of five threads creating and scheduling tasks using the SyncScheduler and a simple FIFO scheduling policy.
$Th_0$ has already created and inserted three tasks ($T_0$ - $T_3$) that are inside the SPSC queue before creating and inserting four additional tasks ($T_4$ - $T_7$). $Th_1$ to $Th_4$ call the {\it getReadyTask} function, one after the other, to obtain a ready task.
The call to {\it lockOrDelegate} of $Th_1$ acquires the lock, while the other threads delegate and busy-wait. Once $Th_1$ is inside the lock, it calls {\it processReadyTasks} and inserts tasks $T_0$ to $T_3$ into the actual scheduler. Then, $Th_1$ schedules one ready task for each of the waiting threads, and finally, it gets a ready task for itself. When $Th_3$ finishes the execution of $T_1$, it calls again to {\it getReadyTask}, and just after that, $Th_2$ does the same. $Th_3$ acquires the lock and calls {\it processReadyTasks}, moving the tasks from the SPSC queue ($T_4$ and $T_5$) to the actual scheduler. Finally, $Th_3$ schedules $T_4$ for $Th_2$, and then, it executes $T_5$.

In microbenchmarks, we found a fourfold speedup on task scheduling using a DTLock compared to a PTLock, and a twelvefold speedup compared to serial task insertion thanks to the SPSC queues.


\section{Memory management}

When optimizing a runtime to achieve the lowest overhead, every operation that requires synchronization between thre\-ads quickly becomes a bottleneck.
This is the case for memory allocation.
Some general-purpose allocators are not well suited to handle a high volume of memory requests in many-core systems.
Many implementations require the serialization of every allocation in the system.
Additionally, the operating system may introduce even more overhead when allocators request more memory areas through system calls.

In the Nanos6 runtime case, removing contention from the scheduler and dependencies caused an even more significant bottleneck on memory allocation.
However, the current state of the art techniques for scalable memory allocation can be applied to any software \cite{hoard, jemalloc}, solving most of the contention problems.
To solve this bottleneck and achieve the performance presented in this article, we had to substitute the default allocator in Nanos6 for Jemalloc, a widely used scalable memory allocator.


%


\section{Instrumentation}
Analyzing runtime performance and finding problems or bottlenecks in the different component requires a mechanism to collect fine-grained instrumentation data.
This instrumentation must have a very low overhead, which is difficult to achieve on a very optimized runtime.
Additionally, runtime systems are sensitive to OS noise (such as thread preemptions), making exploiting kernel internals particularly useful when evaluating latency-critical features, such as those presented in this article.
For this reason, we have developed a new tracing backend aiming at minimum overhead and with both runtime and kernel tracing capabilities.

The backend generates traces in the Common Trace Format (CTF) \cite{ctf}, which strives for fast data writes.
Instrumentation overhead is minimized by storing events on lock-free NUMA-aware per-core circular buffers.
Each buffer is divided into page-aligned sub-buffers that, when full, are periodically flushed to a tmpfs backed file by Nanos6 threads between tasks execution.
Each file contains a time-ordered event subset of the final trace, with either kernel or user events.


Nanos6 threads write events on the lock-free per-core buffer they are pinned to.
User-selected Kernel events are obtained from a per-core memory-mapped circular buffer exported by the Linux Kernel through the {\it perf\_event\_open()} system call.
Between task executions, Nanos6 threads read and format the kernel events according to the CTF specification and move them to an exclusive kernel-events Nanos6 per-core circular buffer.


\section{Evaluation}
In this section, we evaluate the effects of the different optimizations that we present in this paper.
To evaluate their impact, we have prepared different versions of the Nanos6 runtime system, where each one removes one of the three optimizations.
This methodology allows for a better understanding of which optimizations have the most significant impact on runtime performance.
To prove that these optimizations make Nanos6 one of the lowest-overhead task runtimes, we also compare our most optimized version with the most relevant OpenMP implementations.
We conduct our experiments on three HPC machines.

\subsection{Methodology}

To evaluate the task-based runtimes and check the capability of scaling to more finely partitioned work, we will use the following benchmarks, running constant problem sizes and varying the task granularity. These are (1) a \textbf{Dot product} between two arrays that uses a task reduction to aggregate the results from each block, (2) an iterative \textbf{Gauss-Seidel} method solving the heat equation of a 2-D matrix in blocks and task reductions to calculate the residual of each time step, (3) a taskified \textbf{HPCCG} with with several kernels using task reductions and multi-dependencies, (4) a taskified version of \textbf{Lulesh} 2.0 \cite{lulesh}, (5) a taskified \textbf{miniAMR} that mimics the different patterns of Adaptive Mesh Refinement applications \cite{miniAMR, sala2020towards}, (6) a classic parallel blocked \textbf{Matmul}, (7) an \textbf{NBody} benchmark that mimics dynamic particle system simulations, and (8) a blocked \textbf{Cholesky} decomposition that is generally compute bound.

We ran our experiments on various HPC platforms featuring very different architectures:
(1) the \textbf{Intel Xeon} with 2x Intel Xeon Platinum 8160 (Sky-lake) for a total of 48 cores,
(2) the \textbf{ARM Graviton2} with 64 ARM Neoverse N1 cores, and
(3) the \textbf{AMD Rome} with 2x AMD EPYC 7H12 processors for a total of 128 cores and 256 hardware threads.

For all benchmarks, the parallelization is implemented using tasks with OpenMP and OmpSs-2 versions that feature the same amount of parallelism.
However, the kernels used inside each task are sourced from the best available vendor library for each machine, to guarantee competitive performance.
In the AMD and Intel machines this library was Intel MKL, and on the ARM Graviton2 we used the ARM Performance Libraries.
We ran each benchmark a minimum of five times to extract each measurement.

Following this evaluation, we will compare our optimized Nanos6 runtime with the most relevant OpenMP implementations for each machine, including GOMP 9.2.0 \cite{GOMPSource}, LLVM 10 \cite{LLVMSource}, Intel OpenMP and the AMD AOCC depending on availability for each platform.
It is worth noting that both the LLVM, AMD AOCC and Intel OpenMP runtime are based on a work-stealing scheduler, which will allow us to determine if our centralized delegation-based implementation can outperform work-stealing runtimes.

Note that some combinations might not be available depending on the platform because of incompatibilities, non-implemented OpenMP features or compiler bugs. For brevity, we only show the four most relevant benchmarks for each machine.

\subsection{Results}

The best way we found to evaluate objectively how each component of the runtime affects the overall performance is to remove the optimizations we have described selectively. 
To present the results, we use a metric \cite{taskbench} we will refer to as \textit{efficiency}.
It is calculated by dividing the performance of a specific run of a benchmark by the peak performance obtained across all executions. 
This \textit{efficiency} provides a view of how close to peak performance is a specific run while being agnostic to benchmark specific units.
Combining this metric with varying task granularity \cite{OpenMPGranularities, WorksharingMarcos} gives a good view of each runtime version's scalability.
The granularity is expressed in instructions executed per task, which gives an approximation of the task's size.
We chose this unit instead of using time or cycles because the scheduling policies used by each of the runtimes can affect the execution time of a task, and thus it could result in unfair comparisons.

The runtime version without the wait-free dependencies uses the previous dependency implementation based on fine-grained locking.
The variant without the DTLock has a simple mutual exclusion mechanism (based on the PTLock) protecting the scheduler.
Finally, the version without jemalloc uses the standard system allocator for each of the machines.

\begin{figure}
	\centering
	\subcaptionbox*{}{
		\includegraphics[width=\columnwidth]{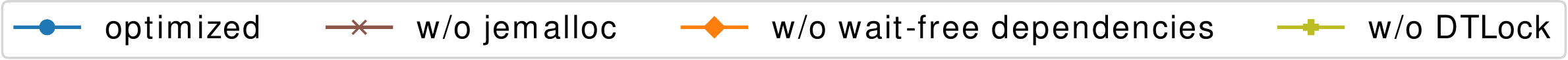}
	}\\[-2.5ex]
	\setlength{\tabcolsep}{0pt}
	\subcaptionbox*{}{
		\begin{tabular}{p{2.5mm} c}
			\multirow{2}{*}{\hspace{-1mm}\rotatebox[origin=lc]{90}{\footnotesize{Efficiency}}} &
			\includegraphics[width=.45\columnwidth]{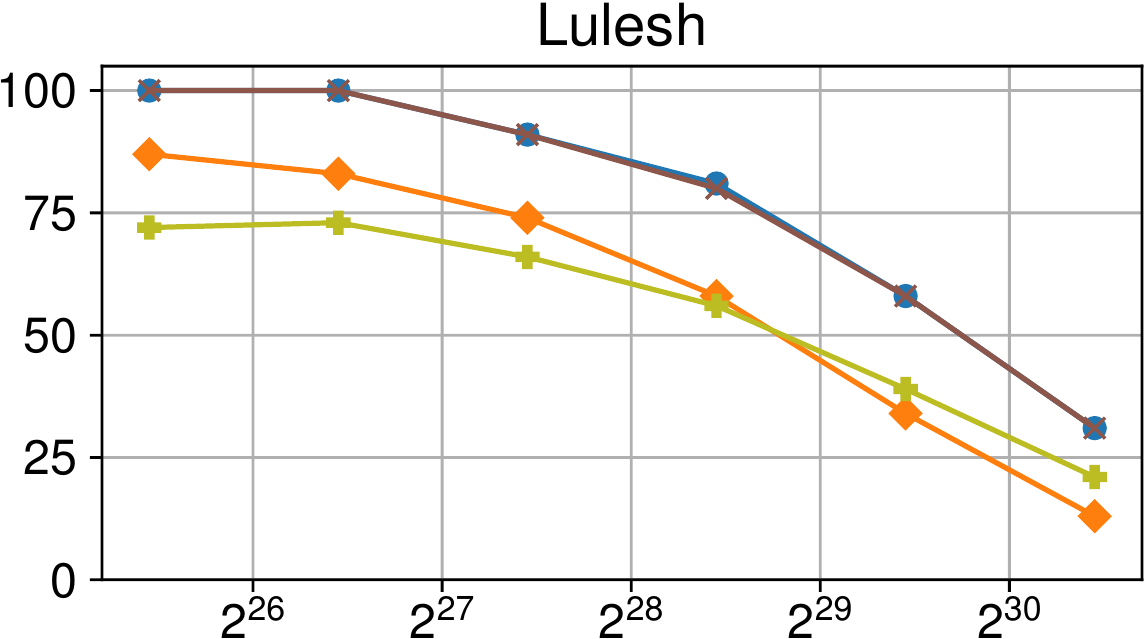}
			\\
			&\includegraphics[width=.45\columnwidth]{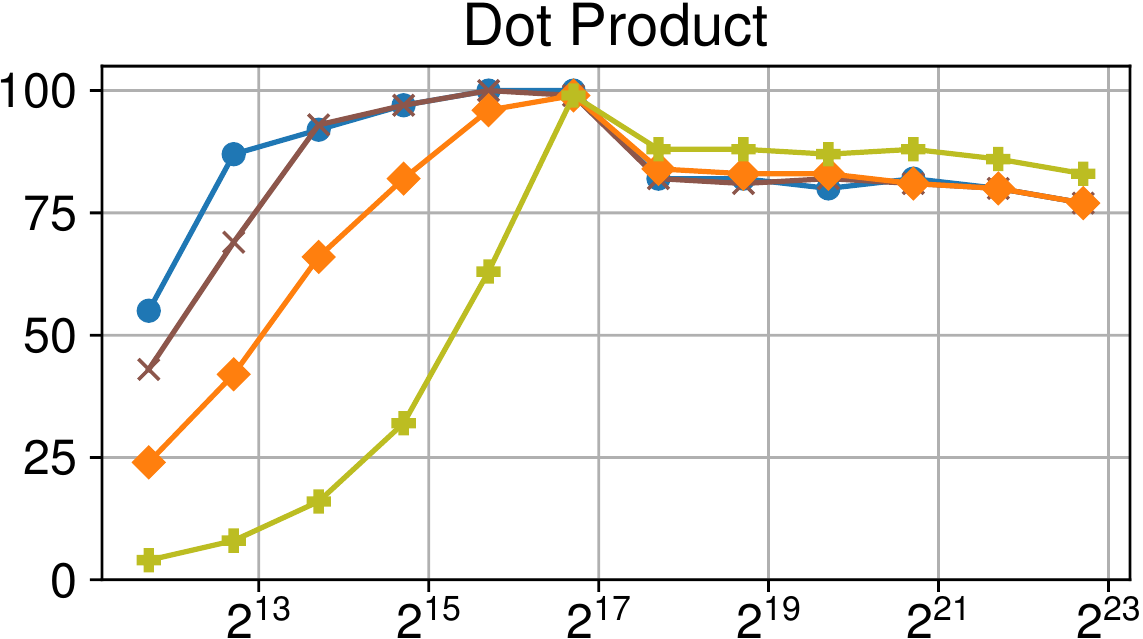}
			\\
			& \multicolumn{1}{r}{\footnotesize{Task}}
		\end{tabular}
	}
	\subcaptionbox*{}{
		\begin{tabular}{c}
			\includegraphics[width=.45\columnwidth]{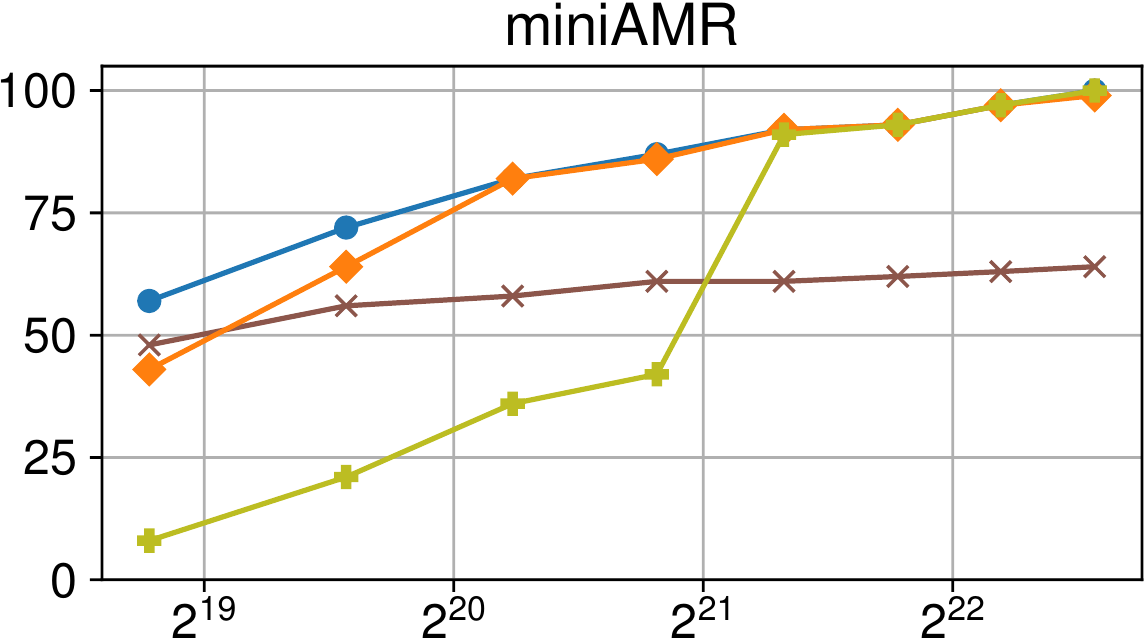}
			\\
			\includegraphics[width=.45\columnwidth]{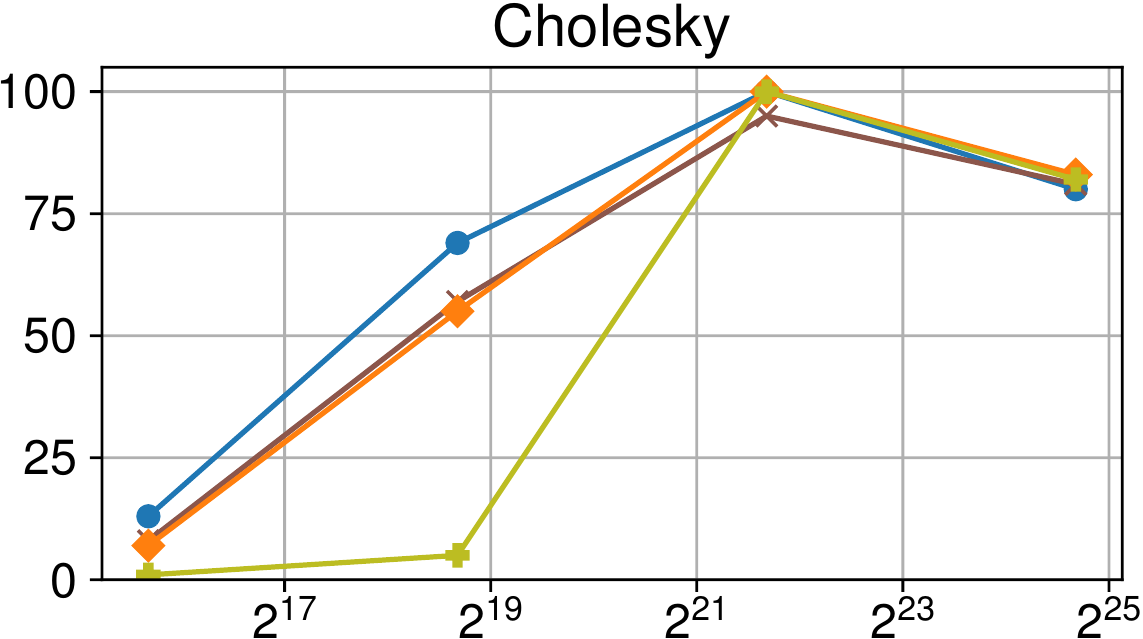}
			\\
			\multicolumn{1}{l}{\footnotesize{granularity}}
		\end{tabular}
	}
	\vspace{-2em}
	\caption{Efficiency vs task granularity of the Nanos6 runtime with and without the described optimizations on Intel Xeon (higher is better)}
	\Description{Efficiency vs task granularity of the Nanos6 runtime with and without the described optimizations on Intel Xeon (higher is better)}
	\label{fig:results-nanos}
\end{figure}

Figure \ref{fig:results-nanos} displays our benchmarks running in the Intel Xeon platform.
The different versions allow us to explore precisely how each optimization affects the scalability for different benchmarks.
The results also confirm that for every benchmark at least one optimization greatly increases the performance for fine-grained tasks.

\begin{figure}
	\centering
	\subcaptionbox*{}{
		\includegraphics[width=\columnwidth]{results/Legend.pdf}
	}\\[-2.5ex]
	\setlength{\tabcolsep}{0pt}
	\subcaptionbox*{}{
		\begin{tabular}{p{2.5mm} c}
			\multirow{2}{*}{\hspace{-1mm}\rotatebox[origin=lc]{90}{\footnotesize{Efficiency}}} &
			\includegraphics[width=.45\columnwidth]{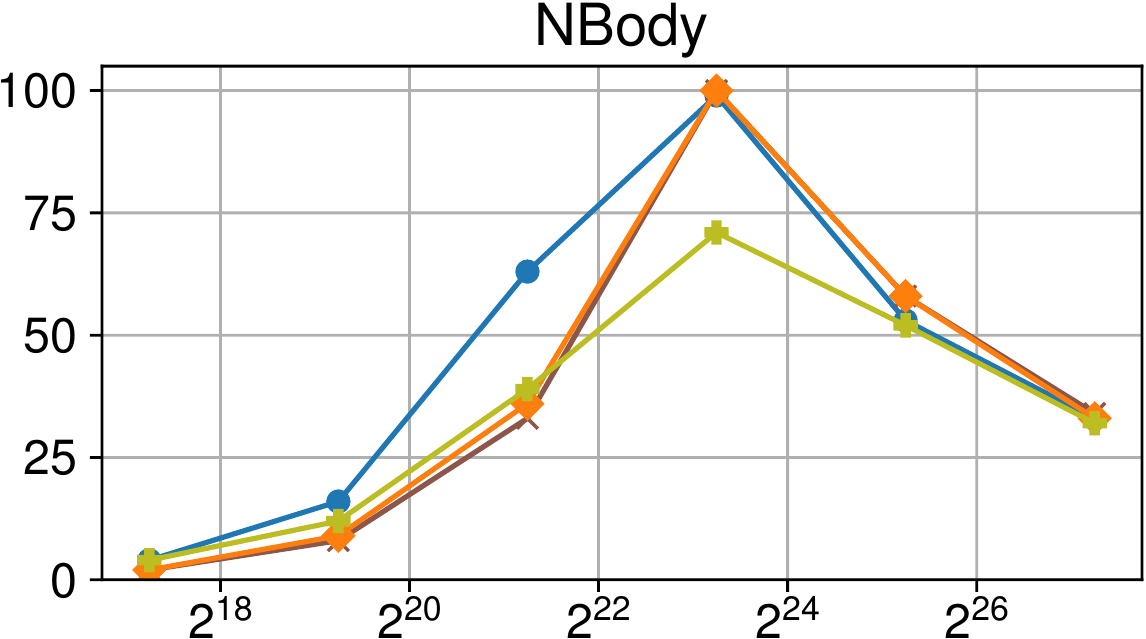}
			\\
			&\includegraphics[width=.45\columnwidth]{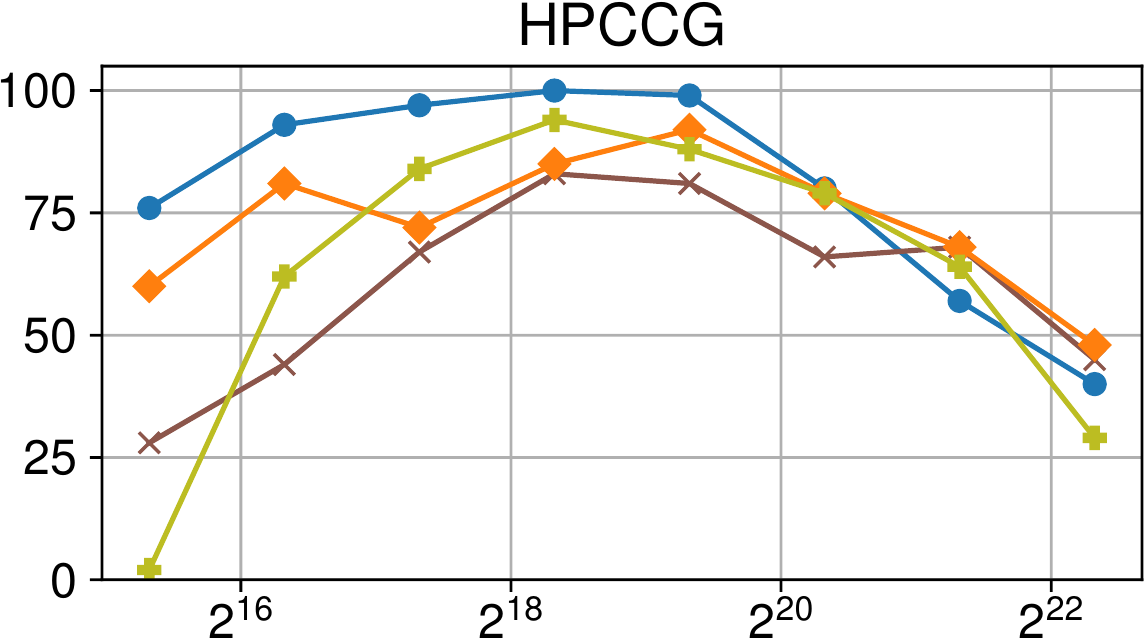}
			\\
			& \multicolumn{1}{r}{\footnotesize{Task}}
		\end{tabular}
	}
	\subcaptionbox*{}{
		\begin{tabular}{c}
			\includegraphics[width=.45\columnwidth]{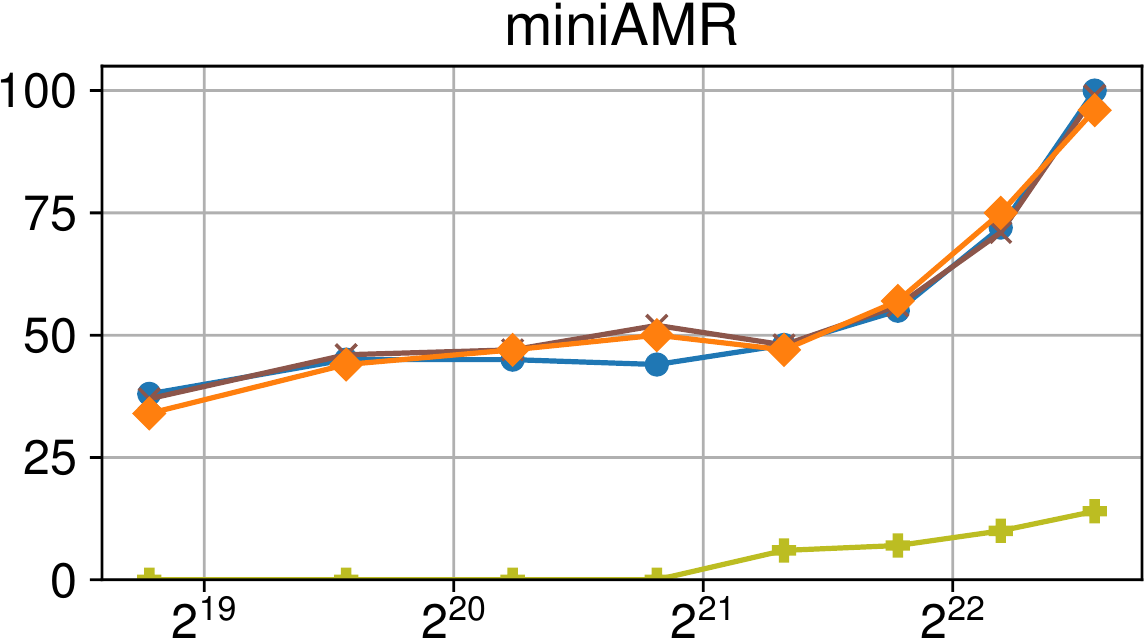}
			\\
			\includegraphics[width=.45\columnwidth]{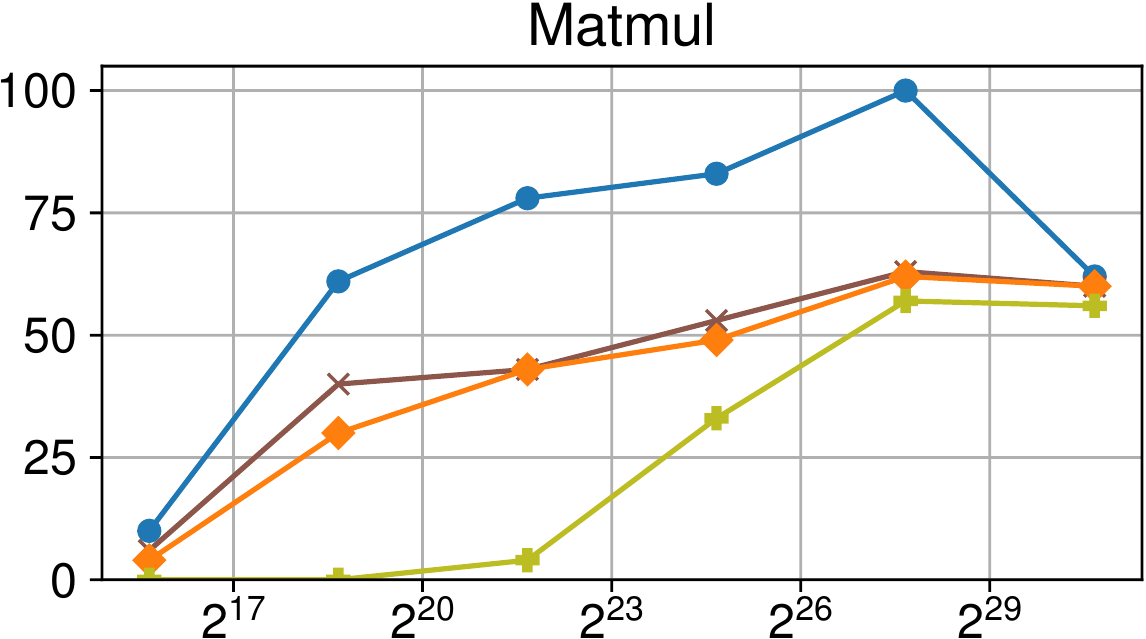}
			\\
			\multicolumn{1}{l}{\footnotesize{granularity}}
		\end{tabular}
	}
	\vspace{-2em}
	\caption{Efficiency vs task granularity of the Nanos6 runtime with and without the described optimizations on AMD Rome (higher is better)}
	\Description{Efficiency vs task granularity of the Nanos6 runtime with and without the described optimizations on AMD Rome (higher is better)}
	\label{fig:results-nanos-amd}
\end{figure}

Figure \ref{fig:results-nanos-amd} shows a similar picture on the AMD Rome system.
This system has a much larger number of CPUs, which can increase the performance degradation caused by heavily contended locks.
Illustrating this point, we see that the scheduler optimization is much more relevant than in the Intel Xeon.
The clearest example is seen on the \textit{miniAMR} benchmark, which we analyze with detailed traces in subsection \ref{sec:traces}.

\begin{figure}
	\centering
	\subcaptionbox*{}{
		\includegraphics[width=\columnwidth]{results/Legend.pdf}
	}\\[-2.5ex]
	\setlength{\tabcolsep}{0pt}
	\subcaptionbox*{}{
		\begin{tabular}{p{2.5mm} c}
			\multirow{2}{*}{\hspace{-1mm}\rotatebox[origin=lc]{90}{\footnotesize{Efficiency}}} &
			\includegraphics[width=.45\columnwidth]{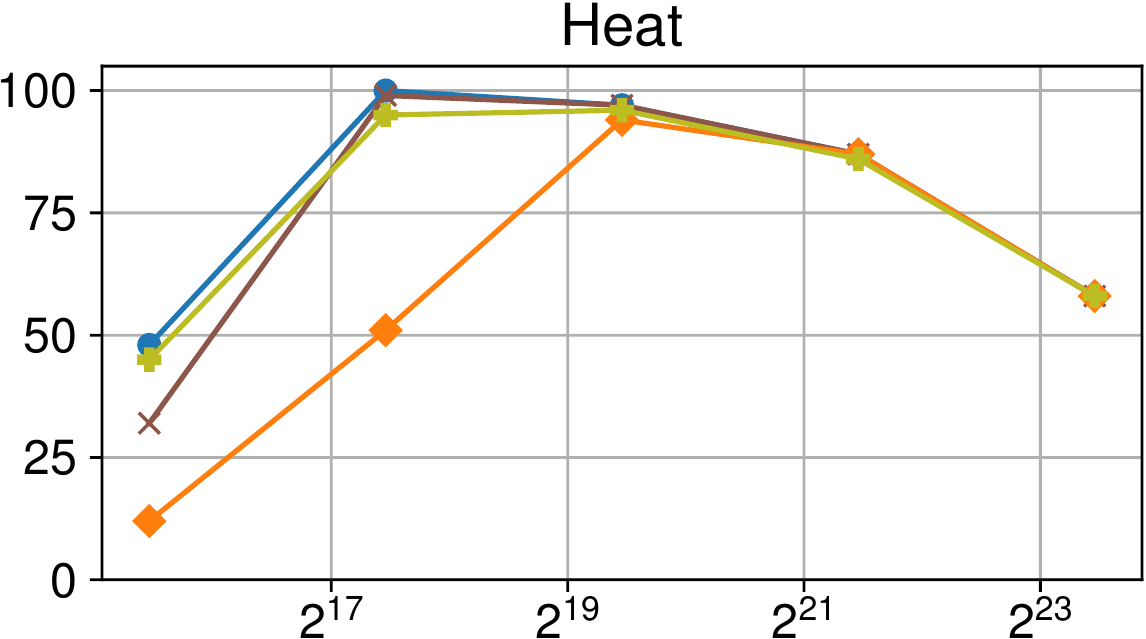}
			\\
			&\includegraphics[width=.45\columnwidth]{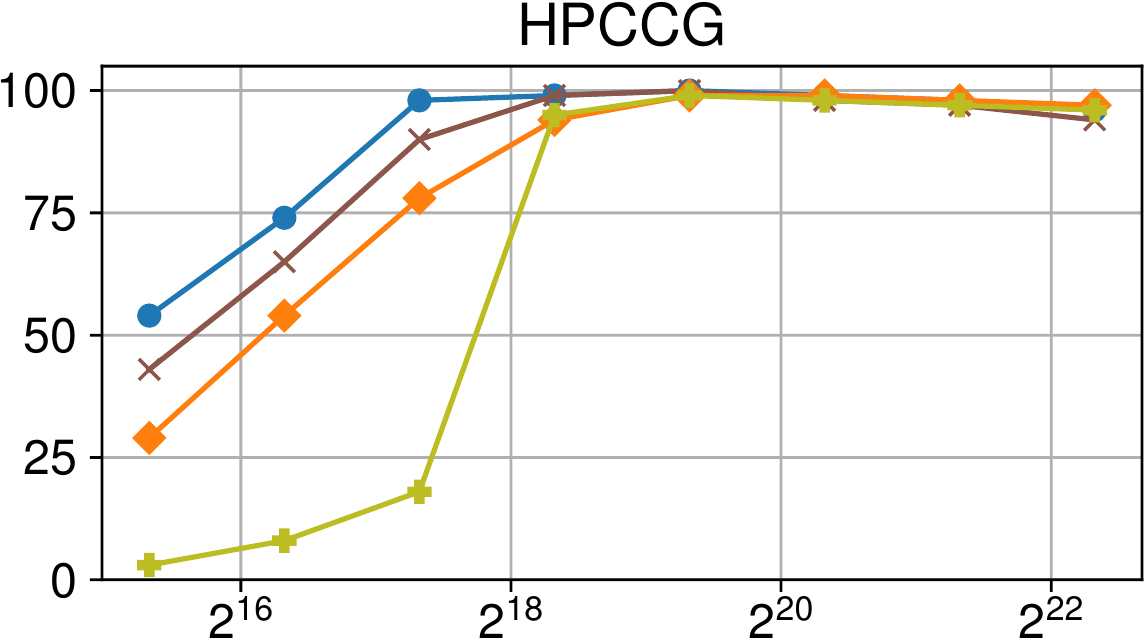}
			\\
			& \multicolumn{1}{r}{\footnotesize{Task}}
		\end{tabular}
	}
	\subcaptionbox*{}{
		\begin{tabular}{c}
			\includegraphics[width=.45\columnwidth]{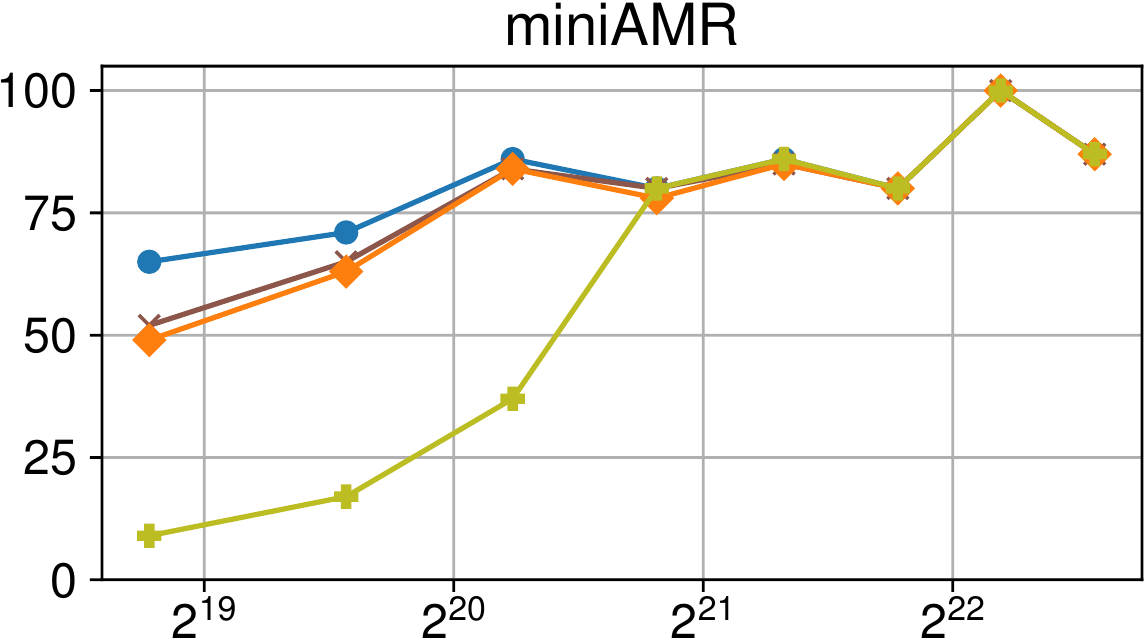}
			\\
			\includegraphics[width=.45\columnwidth]{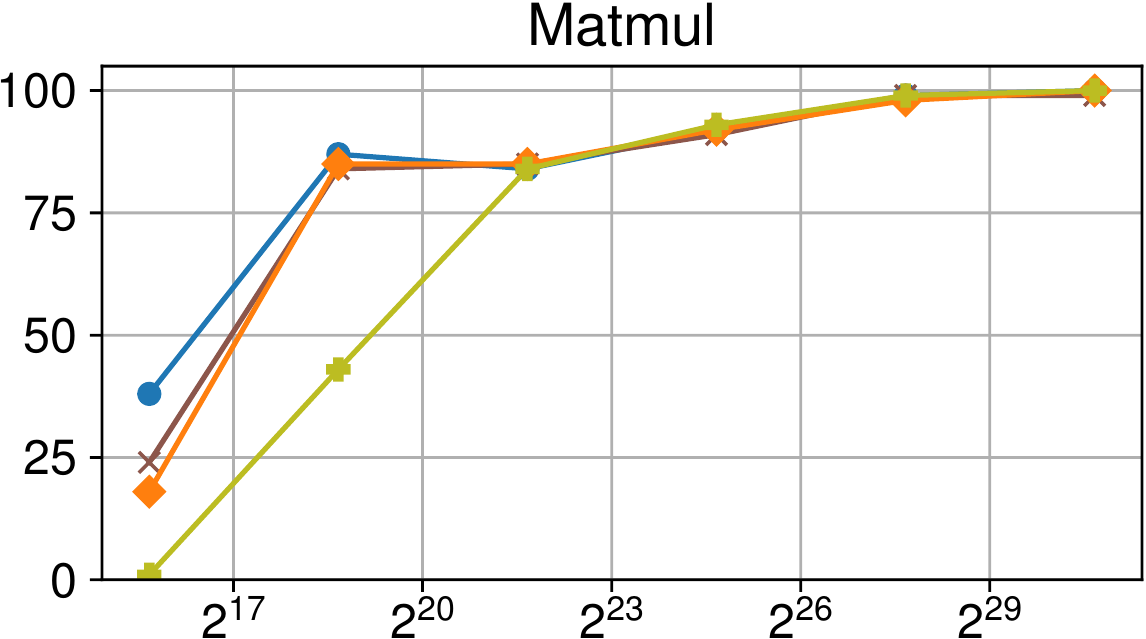}
			\\
			\multicolumn{1}{l}{\footnotesize{granularity}}
		\end{tabular}
	}
	\vspace{-2em}
	\caption{Efficiency vs task granularity of the Nanos6 runtime with and without the described optimizations on ARM Graviton2 (higher is better)}
	\Description{Efficiency vs task granularity of the Nanos6 runtime with and without the described optimizations on ARM Graviton2 (higher is better)}
	\label{fig:results-nanos-arm}
\end{figure}

Finally, Figure \ref{fig:results-nanos-arm} shows the same benchmarks running on an ARM Graviton2 system.
Results are similar to our Intel Xeon evaluation, although some benchmarks have different behaviors due to the lack of NUMA effects on this platform.

Overall, we have seen that our optimizations achieve significant performance gains, especially on small task granularities.
Depending on the benchmark, the wait-free dependencies or the scheduler are the most important optimizations.
However, the scalable memory allocator also delivers some notable performance improvements, especially on the Intel Xeon and AMD Rome machines.

\subsection{Comparison versus OpenMP}

To give context to the results, we compare the optimized Nanos6 runtime to current state-of-the-art OpenMP runtimes on the same three systems.
Our baseline for every benchmark will be the GOMP runtime, distributed with the GCC Compiler, and the LLVM OpenMP Runtime.
However, on Intel Xeon we use the Intel OpenMP runtime, and on AMD Rome the runtime provided by the AMD AOCC, as long as they implement all the features needed by the benchmark.
To ensure a fair comparison, we expressed the same parallelism in the OmpSs-2 and OpenMP versions of the benchmarks.
We also used available kernels in Intel MKL or the ARM Performace Libraries, to prevent noise introduced by the compilers to alter the results.

\begin{figure}
	\centering
	\subcaptionbox*{}{
		\includegraphics[width=.8\columnwidth]{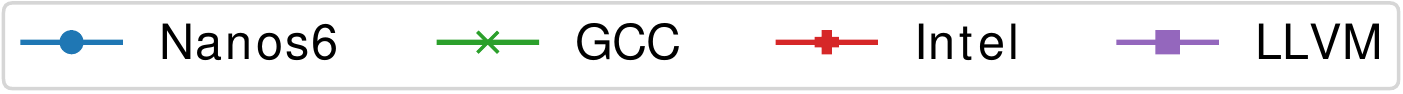}
	}\\[-2.5ex]
	\setlength{\tabcolsep}{0pt}
	\subcaptionbox*{}{
		\begin{tabular}{p{2.5mm} c}
			\multirow{2}{*}{\hspace{-1mm}\rotatebox[origin=lc]{90}{\footnotesize{Efficiency}}} &
		\includegraphics[width=.45\columnwidth]{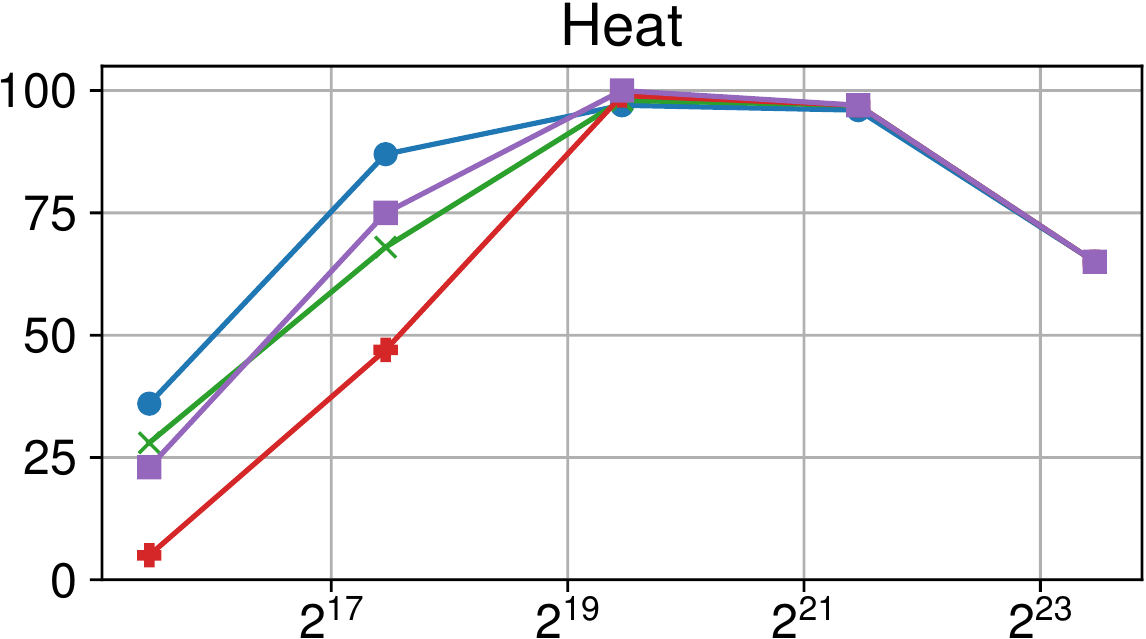}
		\\
		&\includegraphics[width=.45\columnwidth]{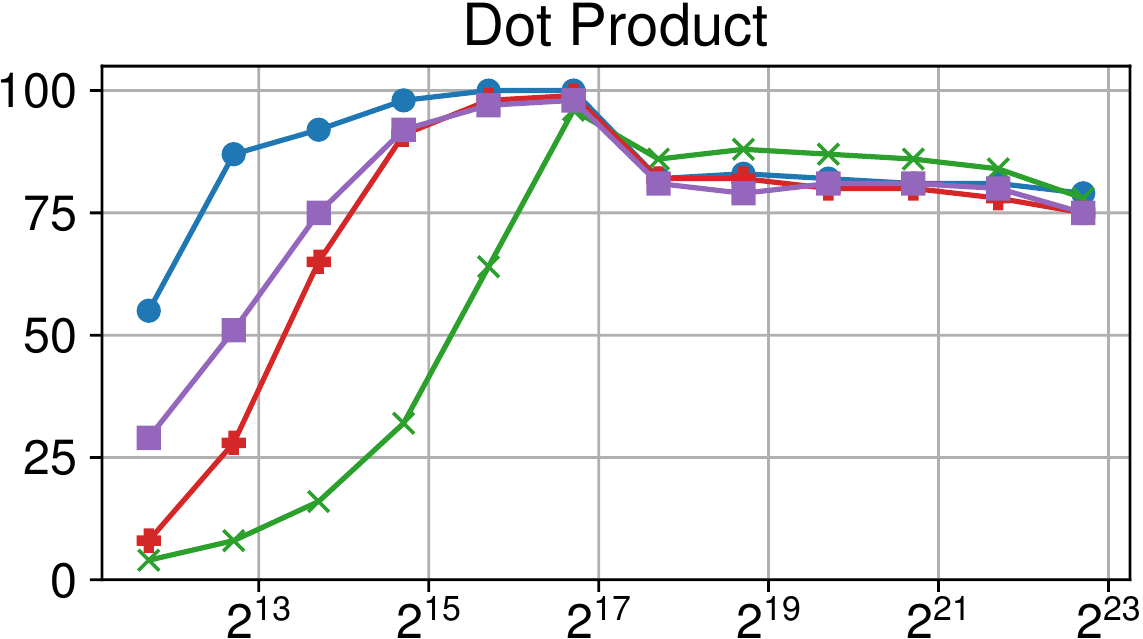}
		\\
		& \multicolumn{1}{r}{\footnotesize{Task}}
		\end{tabular}
	}
	\subcaptionbox*{}{
		\begin{tabular}{c}
		\includegraphics[width=.45\columnwidth]{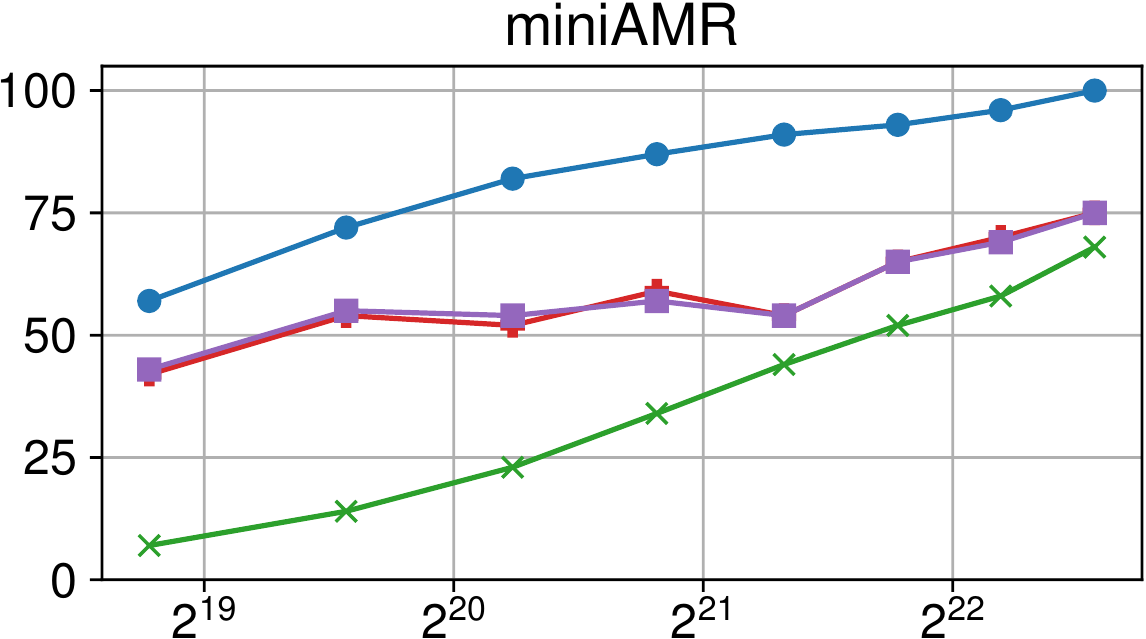}
		\\
		\includegraphics[width=.45\columnwidth]{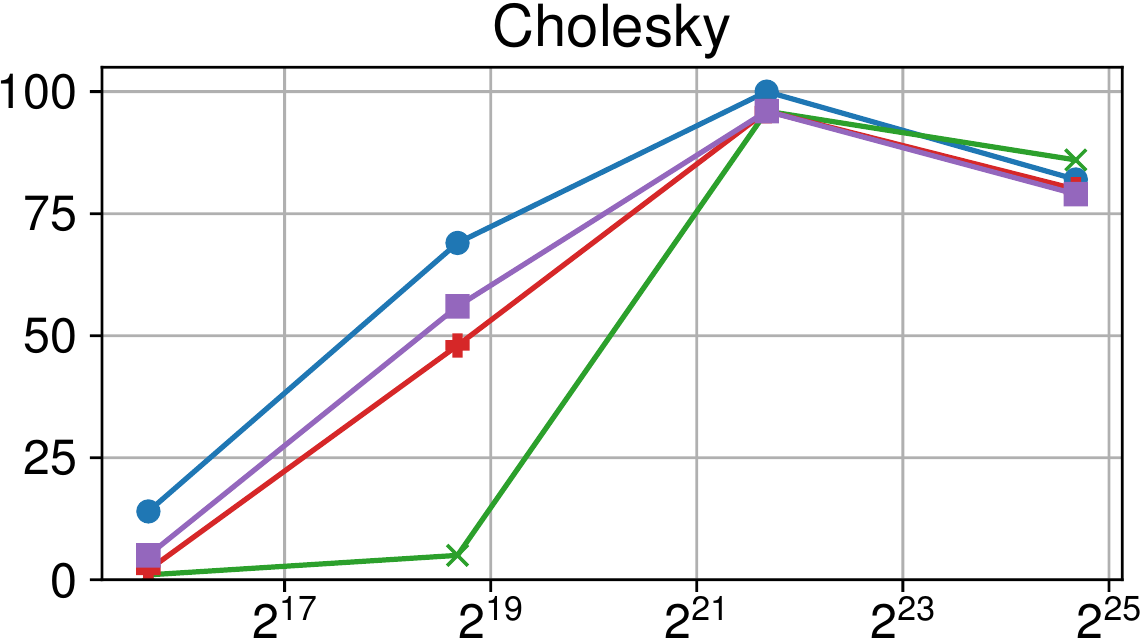}
		\\
		\multicolumn{1}{l}{\footnotesize{granularity}}
		\end{tabular}
	}
	\vspace{-2em}
	\caption{Comparison of performance between the current Nanos6 version and the main OpenMP runtimes on Intel Xeon (higher is better)}
	\Description{Comparison of performance between the current Nanos6 version and the main OpenMP runtimes on Intel Xeon (higher is better)}
	\label{fig:results-omp}
\end{figure}

Figures \ref{fig:results-omp}, \ref{fig:results-omp-amd} and \ref{fig:results-omp-arm} feature the results of comparing the current OpenMP implementation versus the optimized variant of the Nanos6 runtime.
The results are really positive, as in all of the machines, and all of the benchmarks, the best performance in small granularity tasks is provided by the Nanos6 runtime.
In some cases, a higher peak performance is also achieved.
This happens when the ideal block size for a specific benchmark is small enough that performs better in one runtime than another.

As for the other runtimes, the LLVM OpenMP implementation comes second in most benchmarks in terms of scalability, and ties on AMD Rome with the runtime provided with the AOCC compiler.
However, this is expected because the AOCC compiler is based on LLVM 10.

\begin{figure}
	\centering
	\subcaptionbox*{}{
		\includegraphics[width=.8\columnwidth]{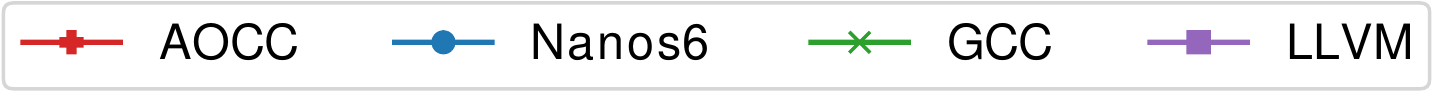}
	}\\[-2.5ex]
	\setlength{\tabcolsep}{0pt}
	\subcaptionbox*{}{
		\begin{tabular}{p{2.5mm} c}
			\multirow{2}{*}{\hspace{-1mm}\rotatebox[origin=lc]{90}{\footnotesize{Efficiency}}} &
		\includegraphics[width=.45\columnwidth]{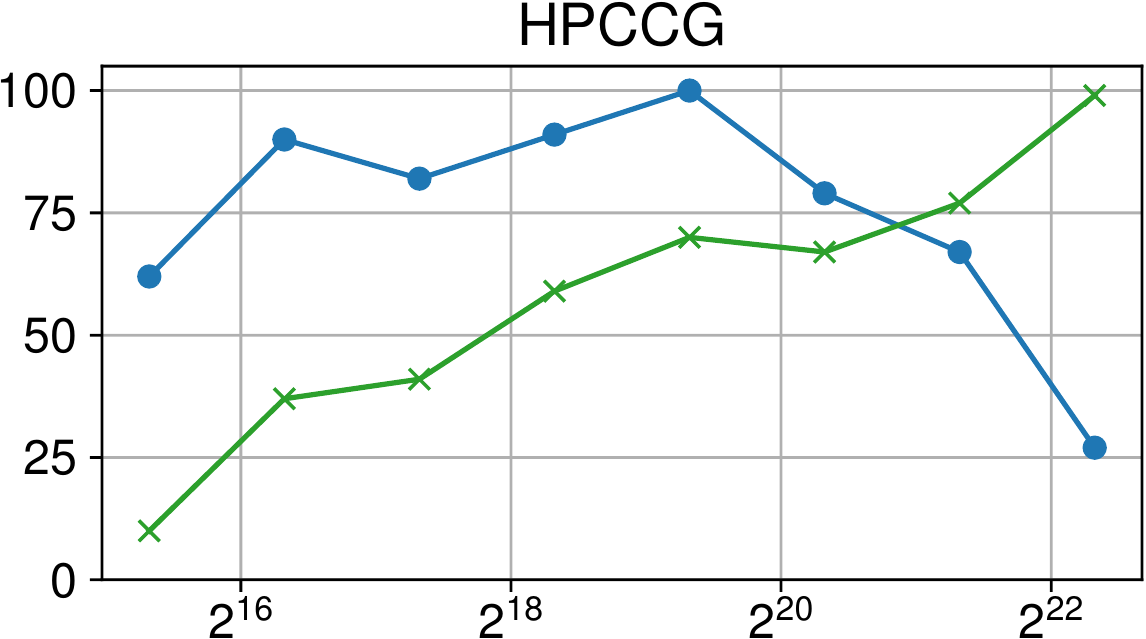}
		\\
		&\includegraphics[width=.45\columnwidth]{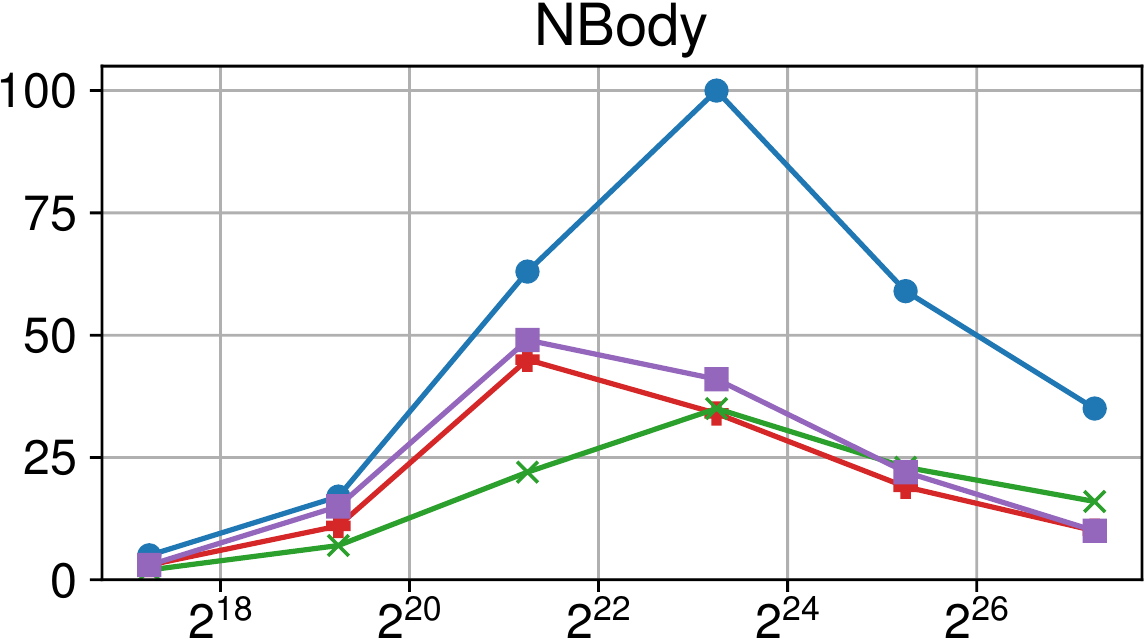}
		\\
		& \multicolumn{1}{r}{\footnotesize{Task}}
		\end{tabular}
	}
	\subcaptionbox*{}{
		\begin{tabular}{c}
		\includegraphics[width=.45\columnwidth]{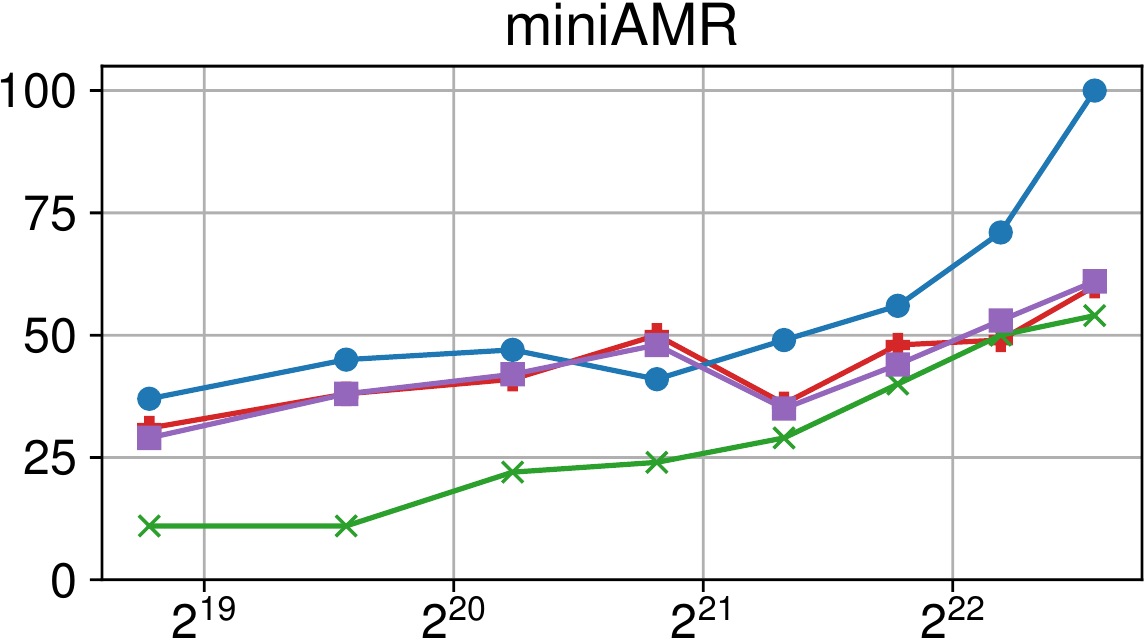}
		\\
		\includegraphics[width=.45\columnwidth]{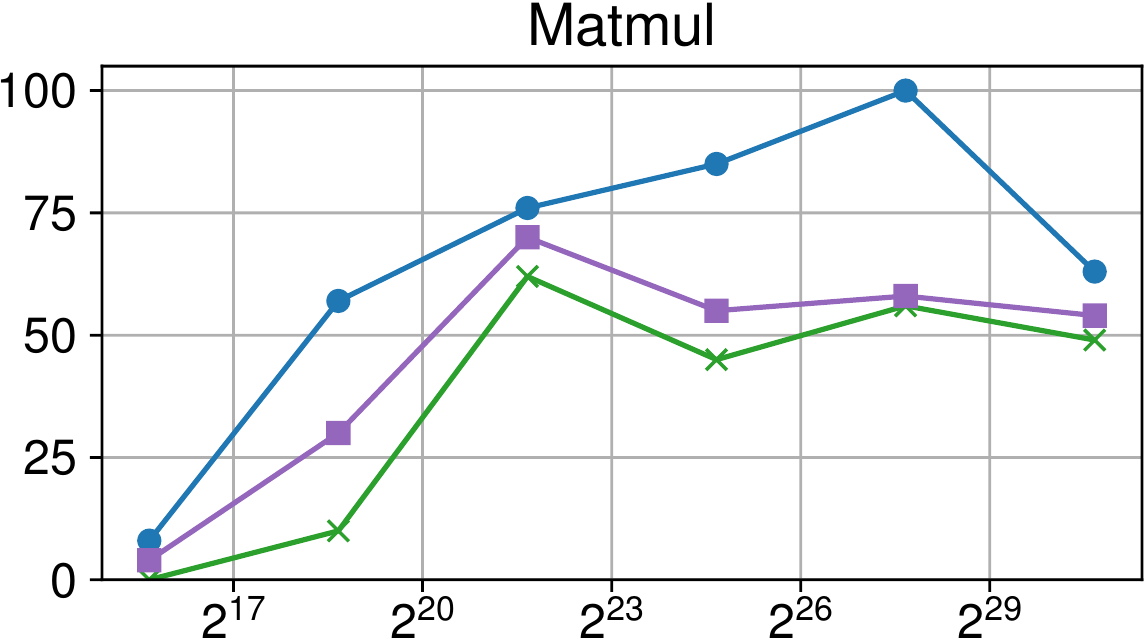}
		\\
		\multicolumn{1}{l}{\footnotesize{granularity}}
		\end{tabular}
	}
	\vspace{-2em}
	\caption{Comparison of performance between the current Nanos6 version and the main OpenMP runtimes on AMD Rome (higher is better)}
	\Description{Comparison of performance between the current Nanos6 version and the main OpenMP runtimes on AMD Rome (higher is better)}
	\label{fig:results-omp-amd}
\end{figure}

\begin{figure}
	\centering
	\subcaptionbox*{}{
		\includegraphics[width=.7\columnwidth]{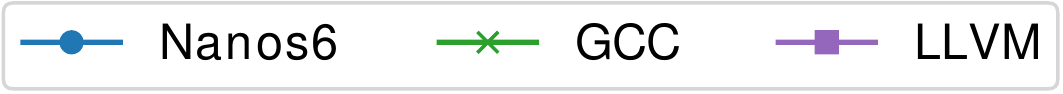}
	}\\[-2.5ex]
	\setlength{\tabcolsep}{0pt}
	\subcaptionbox*{}{
		\begin{tabular}{p{2.5mm} c}
			\multirow{2}{*}{\hspace{-1mm}\rotatebox[origin=lc]{90}{\footnotesize{Efficiency}}} &
		\includegraphics[width=.45\columnwidth]{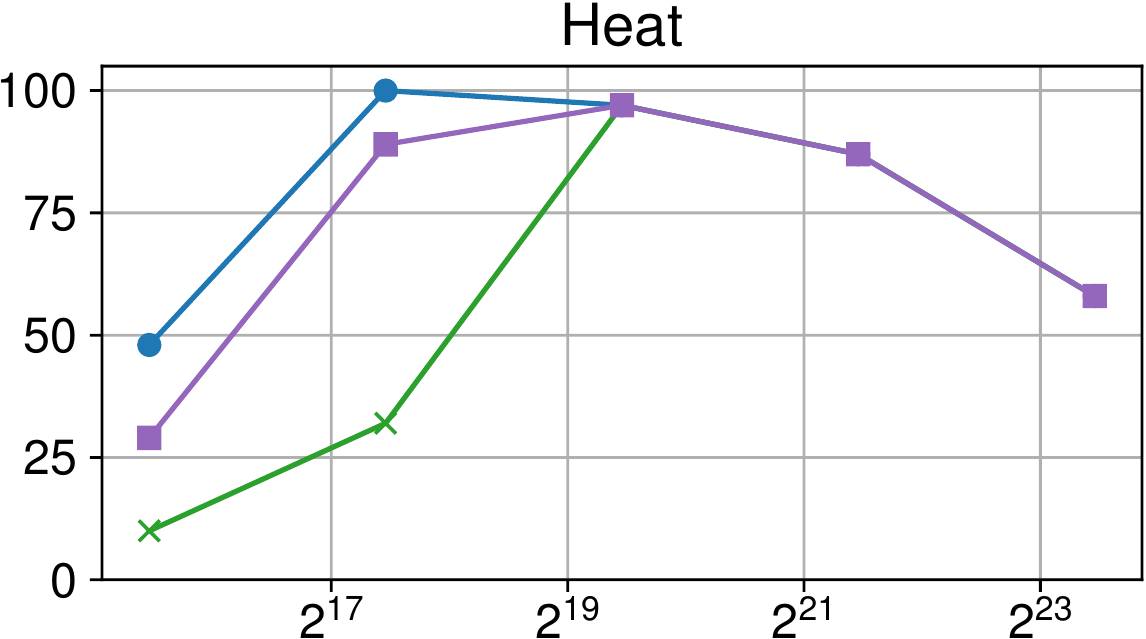}
		\\
		&\includegraphics[width=.45\columnwidth]{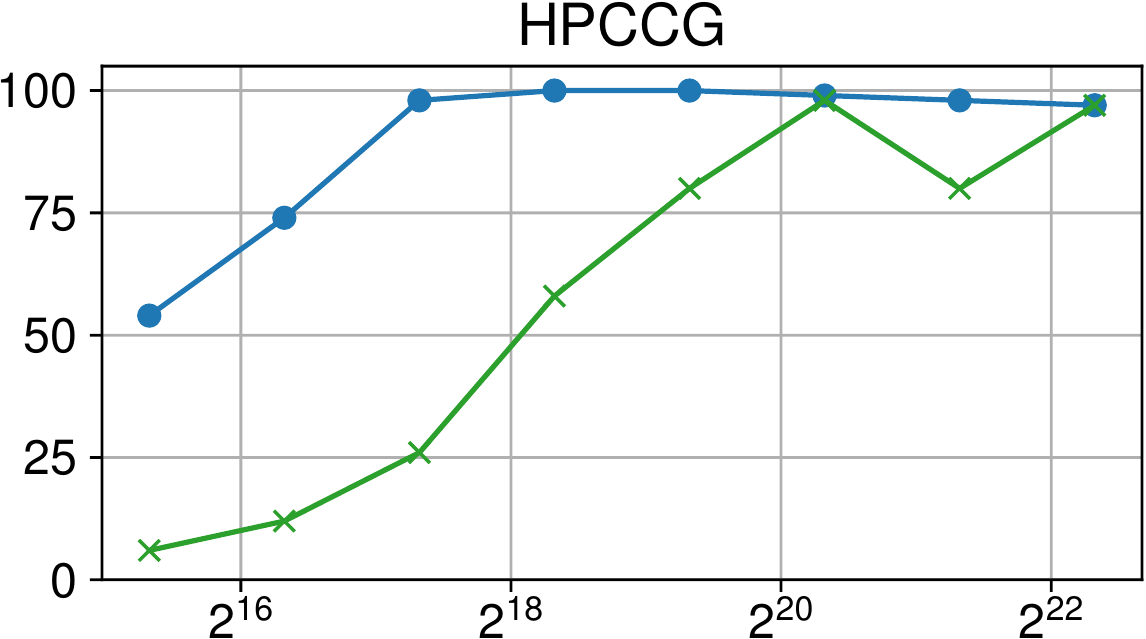}
		\\
		& \multicolumn{1}{r}{\footnotesize{Task}}
		\end{tabular}
	}
	\subcaptionbox*{}{
		\begin{tabular}{c}
		\includegraphics[width=.45\columnwidth]{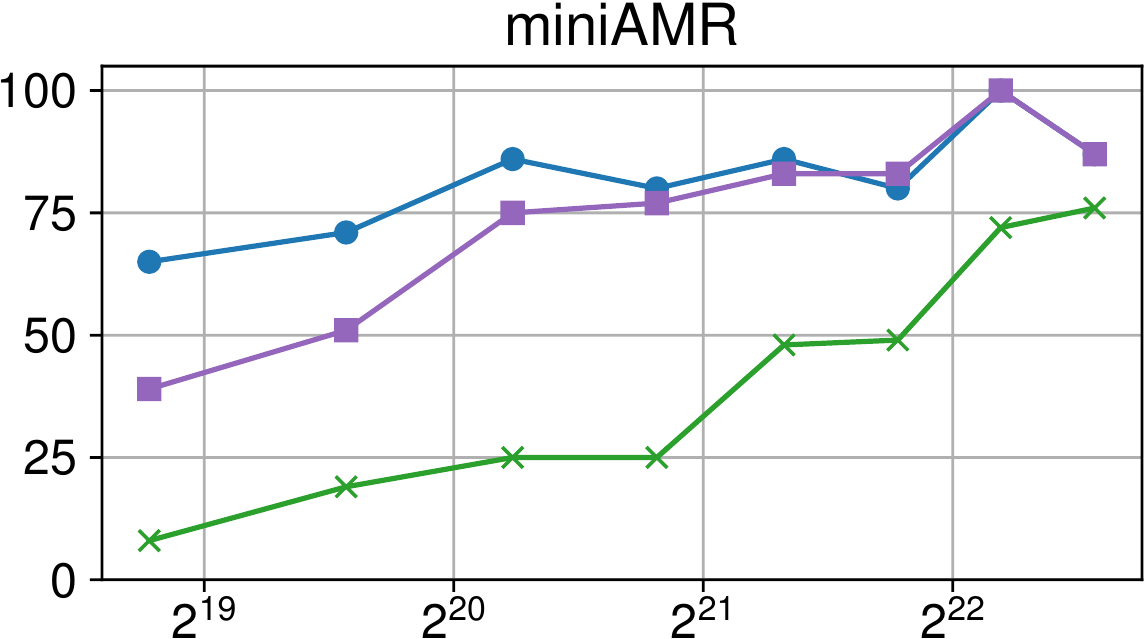}
		\\
		\includegraphics[width=.45\columnwidth]{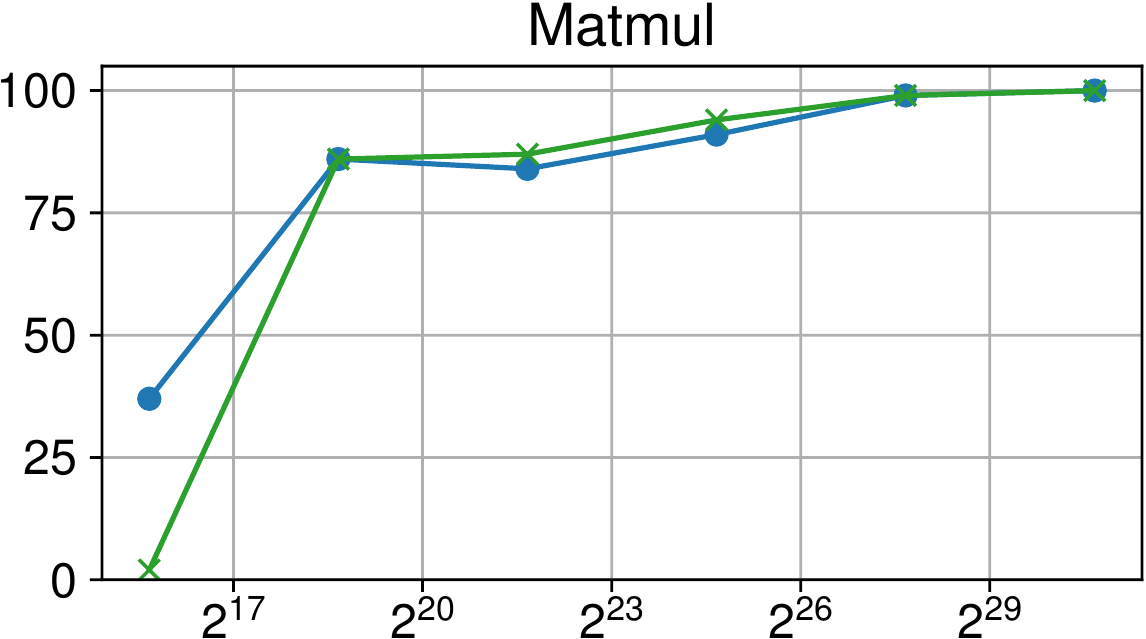}
		\\
		\multicolumn{1}{l}{\footnotesize{granularity}}
		\end{tabular}
	}
	\vspace{-2em}
	\caption{Comparison of performance between the current Nanos6 version and the main OpenMP runtimes on ARM Graviton2 (higher is better)}
	\Description{Comparison of performance between the current Nanos6 version and the main OpenMP runtimes on ARM Graviton2 (higher is better)}
	\label{fig:results-omp-arm}
\end{figure}

\subsection{Detailed traces}
\label{sec:traces}

\begin{figure*}
	\centering
	\includegraphics[width=\textwidth]{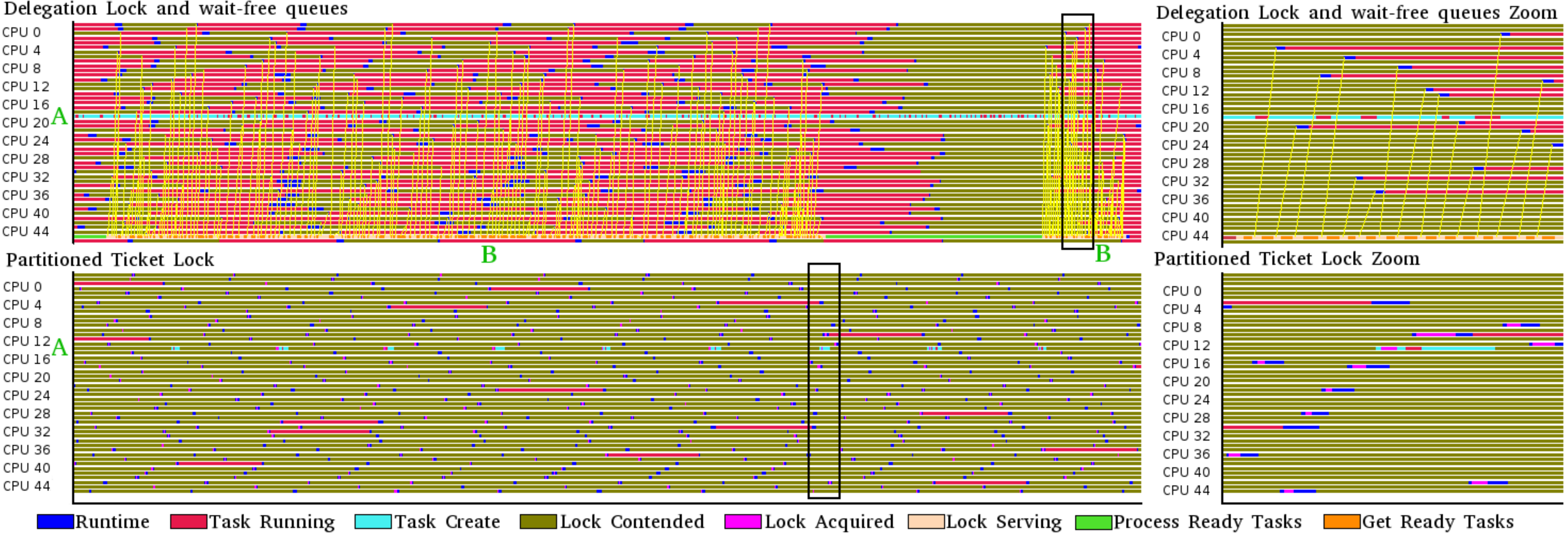}
	\caption{Scheduler lock comparison}
	\Description{Scheduler lock comparison}
	\label{fig:scheduler-trace}
\end{figure*}

Figure \ref{fig:scheduler-trace} shows two miniAMR traces obtained with the new Nanos6 CTF instrumentation backend comparing the Partitioned Ticket Lock (PTlock) and the combination of wait-free queues with the Delegation Ticket Lock (DTLock).
The view displays running tasks (in red), specific runtime subsystems such as task creation (in cyan), or other generic runtime parts (in deep blue) along time (X axis) for a number of cores (Y axis).
The total length is 500us and the zoomed areas (black rectangles) are 10us long, approximately.
Hint {\it A} points to cores running a "task creation" task.
The wait-free version (above) allows created tasks to be queued independently of other cores requesting ready tasks to run.
Instead, in the PTLock version (below), adding and getting a ready task requires obtaining a shared lock which leads the "task creation" task to undergo heavy pressure as other cores requesting a ready task also attempt to acquire the same lock.
Consequently, the number of available ready tasks cannot match the task completion rate and most cores starve (in khaki green).

Hint {\it B} points to DTLock task serving periods.
Yellow arrows depict single tasks being served by the delegation lock owner to waiting cores.
When no more ready tasks left, the lock owner moves ready tasks from the wait-free queues (in green) and continues serving tasks.

\begin{figure}
	\centering
	\includegraphics[width=\columnwidth]{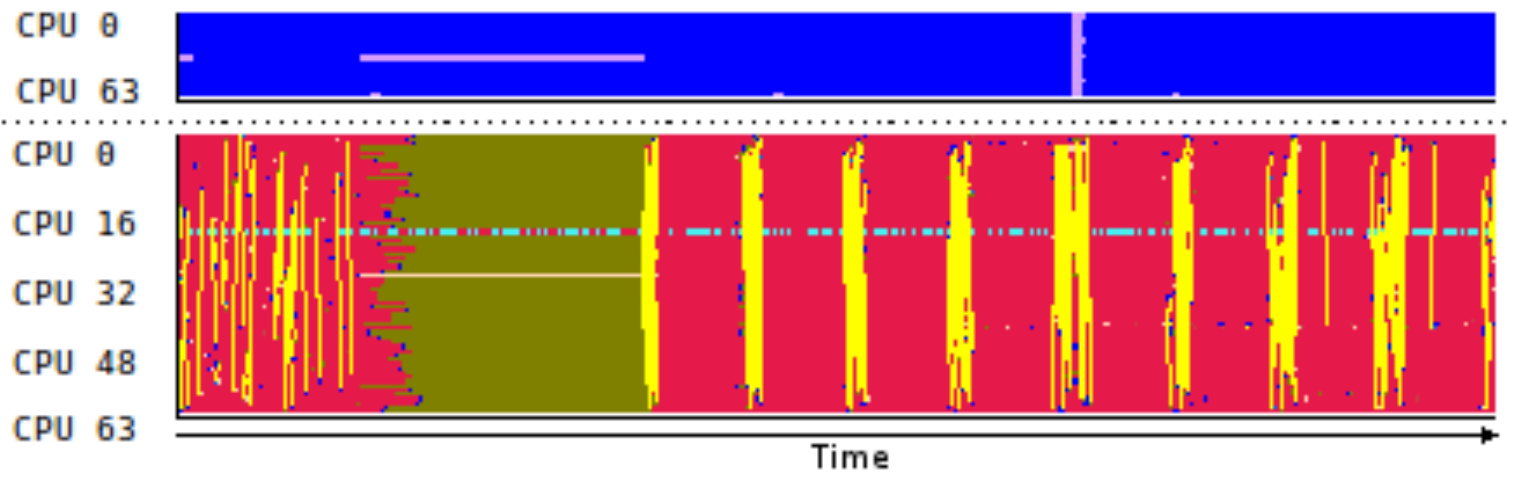}
	\caption{Operating System noise effect on Scheduler}
	\Description{Operating System noise effect on Scheduler}
	\label{fig:scheduler-interrupt}
\end{figure}

Figure \ref{fig:scheduler-interrupt} exemplifies the effect that the operating system noise can incur on the runtime system.
The upper trace displays runtime threads (in deep blue) and hardware interruptions (purple).
The trace below shows a view similar to Figure \ref{fig:scheduler-trace}, where a considerable delay is introduced in the server which causes all cores to stall but the "task creator" core.
Note the yellow lines pattern difference before (irregular) and after (regular) the interrupt.
While the serving thread was stalled in the interrupt, a provision of ready tasks was accumulated.
The surplus of tasks is enough to feed all cores leading to long periods of red (tasks) without yellow lines.
As the extra reserve of tasks lowers, the chances of at least two idle cores requesting a task simultaneously increase, and extra yellow lines appear.
In conclusion, combining OS events with runtime events allows us to complete the whole picture and to identify the source of problems better.

\section{Related work}
\label{sec:relatedwork}

Other dependency system implementations have been described in previous literature, such as the implementation for the OpenUH compiler \cite{PrototypeDependencySupport}.
The GOMP library and LLVM's OpenMP runtime are also available online as Open Source software \cite{GOMPSource}\cite{LLVMSource}.
The topic of overhead in dependency resolution has also been tackled from other angles, such as the TurboBLYSK framework, to create dependency patterns \cite{TurboBLYSK}.

Previous research has also aimed at applying a lock-free approach to dependency resolution, with \cite{DataFlowAnalysis} analyzing several generic dependency resolution schemes and concluding that a ticket-based lockless scheme provided the best performance in their benchmarks.
On the same line, another lock-free dependency system was implemented for the OMPi OpenMP/C Compiler \cite{LocklessListDeps}, which supported OpenMP 4.0 and was based on the same patterns as lock-free lists.
However, our implementation offers a stronger wait-free guarantee.

Regarding task scheduling, there are several studies on alternatives to centralized lock-based scheduling.
Many have studied work-stealing techniques for hierarchical and partitioned schedulers in shared-memory systems \cite{prokopec2013near,guo2010slaw,wang2014adaptive}.
Olivier et al. \cite{olivier2012openmp} proposed a hierarchical scheduler featuring a lock-free ready task queue per socket.
Once the socket's queue is empty, only one of the threads in that queue can try to steal tasks from other sockets, while the others wait.
Similarly, Muddukrishna et al. \cite{muddukrishna2015locality} proposed a lock-based queue per NUMA node, but in this case, workers from the same NUMA can steal at the same time.
In contrast, Vikranth et al. \cite{vikranth2013topology} implemented a task queue per thread, grouped into stealing domains (e.g., one per socket).
Workers always first try to steal from queues in the local domain before trying in the rest.
However, all these approaches can suffer the same bottlenecks as centralized schedulers when running in modern many-core systems.
The ready task queues of consumer threads tend to be empty, so they steal tasks from the creator's queues (usually few), producing contention in those sections.
Also, hierarchical schedulers often have complicated implementations, so developing new scheduling policies becomes an arduous task.

Even though, notice that these hierarchical approaches with work-stealing could use our DTLock (Section \ref{sec:dtlock}) to protect the access to task queues.

The Linux Kernel static tracing infrastructure is used by several backends such as LTTng, ftrace, perf, SystemTap or eBPF.
Yet, only LTTng focuses on efficient user and kernel correlated static tracing \cite{lttng}.
When tracing on the context of runtime systems for HPC, LTTng has two main drawbacks.
On the one hand, tracing the kernel requires installing an out-of-tree Linux Kernel module that usually clashes with operating policies of data centers and supercomputing facilities.
On the other hand, LTTng relies on server daemons to collect and write tracepoints from both user and kernel space.
Such daemons might oversubscribe runtime threads leading to undesired noise.
Therefore, we concluded that our particular case needed an ad-hoc tracing solution.
\section{Conclusion and Future Work}
\label{sec:conclusion}
The proliferation of many-core architectures and workloads with irregular parallelism and load imbalance have shifted the focus
from traditional fork-join parallelism to task-based parallelism.
Nevertheless, task management costs are still an important source of overhead, especially when using fine granularities.
Throughout this paper, we enhance two critical components that bound the ability of runtime systems to manage fine-grained tasks:
the dependency system and the scheduler. We combine both with a state of the art memory allocator to achieve very competitive performance.

We have introduced a novel wait-free approach to implementing dependency management inside a parallel runtime.
We have also defined the \textit{Atomic State Machine} concept and its restrictions and formalized its wait-freedom.
We believe the ASM concept is applicable to similar models and runtimes that use a data-flow execution model.

Additionally, we proposed a novel \textit{Delegation Ticket Lock} that delivers very good performance compared to other
state-of-the-art locks, while keeping the simplicity in the development of scheduling internals and policies.

We also identified the critical contention bottleneck caused by memory management and tackled the problem by leveraging the \textit{jemalloc} state-of-the-art scalable memory allocator.

Finally, we implemented a highly-detailed instrumentation to provide information from both application and kernel level, while introducing
minimal overhead. Such a tool is crucial to identify and analyze bottlenecks in modern runtime systems.

Our evaluation assesses the performance of the different components separately and together, showing important performance improvements compared to
(1) the previous version of the runtime system, and (2) state-of-the-art runtime systems such as Intel OpenMP, GNU GOMP and LLVM OpenMP.

As future work, we plan to investigate extensions of the DTLock interface to support flat combining \cite{10.1145/1810479.1810540}.
This interface will require the ability to access and unblock several waiting threads simultaneously to be able to combine their operations.

\begin{acks}                            

  This project is supported by the European Union's Horizon 2020 Research and Innovation programme under grant agreement No.s 754304 (DEEP-EST), by the Spanish Ministry of Science and Innovation (contract PID2019-107255GB and TIN2015-65316P) and by the Generalitat de Catalunya (2017-SGR-1414). We acknowledge PRACE for awarding us access to Joliot-Curie at GENCI@CEA, France.
\end{acks}

\bibliography{references}


\begin{thebibliography}{39}


\ifx \showCODEN    \undefined \def \showCODEN     #1{\unskip}     \fi
\ifx \showDOI      \undefined \def \showDOI       #1{#1}\fi
\ifx \showISBNx    \undefined \def \showISBNx     #1{\unskip}     \fi
\ifx \showISBNxiii \undefined \def \showISBNxiii  #1{\unskip}     \fi
\ifx \showISSN     \undefined \def \showISSN      #1{\unskip}     \fi
\ifx \showLCCN     \undefined \def \showLCCN      #1{\unskip}     \fi
\ifx \shownote     \undefined \def \shownote      #1{#1}          \fi
\ifx \showarticletitle \undefined \def \showarticletitle #1{#1}   \fi
\ifx \showURL      \undefined \def \showURL       {\relax}        \fi
\providecommand\bibfield[2]{#2}
\providecommand\bibinfo[2]{#2}
\providecommand\natexlab[1]{#1}
\providecommand\showeprint[2][]{arXiv:#2}

\bibitem[\protect\citeauthoryear{{Anastasios Souris}}{{Anastasios
  Souris}}{2015}]%
        {LocklessListDeps}
\bibfield{author}{\bibinfo{person}{{Anastasios Souris}}.}
  \bibinfo{year}{2015}\natexlab{}.
\newblock \bibinfo{title}{Design and implementation of the OpenMP 4.0 task
  dataflow model for cache-coherent shared-memory parallel systems in the
  runtime of the OMPi OpenMP/C compiler}.
\newblock
\newblock
\urldef\tempurl%
\url{https://doi.org/10.26233/HEALLINK.TUC.30331}
\showDOI{\tempurl}


\bibitem[\protect\citeauthoryear{Berger, McKinley, Blumofe, and Wilson}{Berger
  et~al\mbox{.}}{2000}]%
        {hoard}
\bibfield{author}{\bibinfo{person}{Emery~D. Berger},
  \bibinfo{person}{Kathryn~S. McKinley}, \bibinfo{person}{Robert~D. Blumofe},
  {and} \bibinfo{person}{Paul~R. Wilson}.} \bibinfo{year}{2000}\natexlab{}.
\newblock \showarticletitle{Hoard: A Scalable Memory Allocator for
  Multithreaded Applications}. In \bibinfo{booktitle}{\emph{Proceedings of the
  Ninth International Conference on Architectural Support for Programming
  Languages and Operating Systems}} (Cambridge, Massachusetts, USA)
  \emph{(\bibinfo{series}{ASPLOS IX})}. \bibinfo{publisher}{Association for
  Computing Machinery}, \bibinfo{address}{New York, NY, USA},
  \bibinfo{pages}{117–128}.
\newblock
\showISBNx{1581133170}
\urldef\tempurl%
\url{https://doi.org/10.1145/378993.379232}
\showDOI{\tempurl}


\bibitem[\protect\citeauthoryear{Borkar and Chien}{Borkar and Chien}{2011}]%
        {BorkarFuture}
\bibfield{author}{\bibinfo{person}{Shekhar Borkar} {and}
  \bibinfo{person}{Andrew~A. Chien}.} \bibinfo{year}{2011}\natexlab{}.
\newblock \showarticletitle{The Future of Microprocessors}.
\newblock \bibinfo{journal}{\emph{Commun. ACM}} \bibinfo{volume}{54},
  \bibinfo{number}{5} (\bibinfo{date}{May} \bibinfo{year}{2011}),
  \bibinfo{pages}{67–77}.
\newblock
\showISSN{0001-0782}
\urldef\tempurl%
\url{https://doi.org/10.1145/1941487.1941507}
\showDOI{\tempurl}


\bibitem[\protect\citeauthoryear{BSC}{BSC}{2020a}]%
        {bsc2019nanos6}
\bibfield{author}{\bibinfo{person}{BSC}.} \bibinfo{year}{2020}\natexlab{a}.
\newblock \bibinfo{title}{Nanos6 GitHub}.
\newblock
\newblock
\urldef\tempurl%
\url{https://github.com/bsc-pm/nanos6}
\showURL{%
\tempurl}


\bibitem[\protect\citeauthoryear{BSC}{BSC}{2020b}]%
        {bsc2020ompss2}
\bibfield{author}{\bibinfo{person}{BSC}.} \bibinfo{year}{2020}\natexlab{b}.
\newblock \bibinfo{title}{{OmpSs-2} Specification}.
\newblock
\newblock
\urldef\tempurl%
\url{https://pm.bsc.es/ftp/ompss-2/doc/spec/OmpSs-2-Specification.pdf}
\showURL{%
\tempurl}


\bibitem[\protect\citeauthoryear{{Dagum} and {Menon}}{{Dagum} and
  {Menon}}{1998}]%
        {openmpindustry}
\bibfield{author}{\bibinfo{person}{L. {Dagum}} {and} \bibinfo{person}{R.
  {Menon}}.} \bibinfo{year}{1998}\natexlab{}.
\newblock \showarticletitle{{OpenMP}: an industry standard API for
  shared-memory programming}.
\newblock \bibinfo{journal}{\emph{IEEE Computational Science and Engineering}}
  \bibinfo{volume}{5}, \bibinfo{number}{1} (\bibinfo{date}{Jan}
  \bibinfo{year}{1998}), \bibinfo{pages}{46--55}.
\newblock
\showISSN{1558-190X}
\urldef\tempurl%
\url{https://doi.org/10.1109/99.660313}
\showDOI{\tempurl}


\bibitem[\protect\citeauthoryear{Desnoyers}{Desnoyers}{2020}]%
        {ctf}
\bibfield{author}{\bibinfo{person}{Mathieu Desnoyers}.}
  \bibinfo{year}{2020}\natexlab{}.
\newblock \bibinfo{title}{The Common Trace Format}.
\newblock
\newblock
\urldef\tempurl%
\url{https://diamon.org/ctf/v1.8.3}
\showURL{%
\tempurl}


\bibitem[\protect\citeauthoryear{Dice}{Dice}{2011}]%
        {10.1145/1989493.1989543}
\bibfield{author}{\bibinfo{person}{David Dice}.}
  \bibinfo{year}{2011}\natexlab{}.
\newblock \showarticletitle{Brief Announcement: A Partitioned Ticket Lock}. In
  \bibinfo{booktitle}{\emph{Proceedings of the Twenty-Third Annual ACM
  Symposium on Parallelism in Algorithms and Architectures}} (San Jose,
  California, USA) \emph{(\bibinfo{series}{SPAA ’11})}.
  \bibinfo{publisher}{Association for Computing Machinery},
  \bibinfo{address}{New York, NY, USA}, \bibinfo{pages}{309–310}.
\newblock
\showISBNx{9781450307437}
\urldef\tempurl%
\url{https://doi.org/10.1145/1989493.1989543}
\showDOI{\tempurl}


\bibitem[\protect\citeauthoryear{Dice and Kogan}{Dice and Kogan}{2019}]%
        {TWA}
\bibfield{author}{\bibinfo{person}{Dave Dice} {and} \bibinfo{person}{Alex
  Kogan}.} \bibinfo{year}{2019}\natexlab{}.
\newblock \showarticletitle{TWA -- Ticket Locks Augmented with a Waiting
  Array}. In \bibinfo{booktitle}{\emph{Euro-Par 2019: Parallel Processing}},
  \bibfield{editor}{\bibinfo{person}{Ramin Yahyapour}} (Ed.).
  \bibinfo{publisher}{Springer International Publishing},
  \bibinfo{address}{Cham}, \bibinfo{pages}{334--345}.
\newblock
\showISBNx{978-3-030-29400-7}


\bibitem[\protect\citeauthoryear{Duran, Perez, Ayguad{\'e}, Badia, and
  Labarta}{Duran et~al\mbox{.}}{2008}]%
        {taskBased2}
\bibfield{author}{\bibinfo{person}{Alejandro Duran}, \bibinfo{person}{Josep~M.
  Perez}, \bibinfo{person}{Eduard Ayguad{\'e}}, \bibinfo{person}{Rosa~M.
  Badia}, {and} \bibinfo{person}{Jesus Labarta}.}
  \bibinfo{year}{2008}\natexlab{}.
\newblock \showarticletitle{Extending the OpenMP Tasking Model to Allow
  Dependent Tasks}. In \bibinfo{booktitle}{\emph{OpenMP in a New Era of
  Parallelism}}, \bibfield{editor}{\bibinfo{person}{Rudolf Eigenmann} {and}
  \bibinfo{person}{Bronis~R. de~Supinski}} (Eds.). \bibinfo{publisher}{Springer
  Berlin Heidelberg}, \bibinfo{address}{Berlin, Heidelberg},
  \bibinfo{pages}{111--122}.
\newblock
\showISBNx{978-3-540-79561-2}


\bibitem[\protect\citeauthoryear{{Esmaeilzadeh}, {Blem}, {Amant},
  {Sankaralingam}, and {Burger}}{{Esmaeilzadeh} et~al\mbox{.}}{2011}]%
        {DarkSiliconMulticore}
\bibfield{author}{\bibinfo{person}{H. {Esmaeilzadeh}}, \bibinfo{person}{E.
  {Blem}}, \bibinfo{person}{R.~S. {Amant}}, \bibinfo{person}{K.
  {Sankaralingam}}, {and} \bibinfo{person}{D. {Burger}}.}
  \bibinfo{year}{2011}\natexlab{}.
\newblock \showarticletitle{Dark silicon and the end of multicore scaling}. In
  \bibinfo{booktitle}{\emph{2011 38th Annual International Symposium on
  Computer Architecture (ISCA)}}. \bibinfo{pages}{365--376}.
\newblock


\bibitem[\protect\citeauthoryear{Evans}{Evans}{2020}]%
        {jemalloc}
\bibfield{author}{\bibinfo{person}{Jason Evans}.}
  \bibinfo{year}{2020}\natexlab{}.
\newblock \bibinfo{title}{jemalloc}.
\newblock
\newblock
\urldef\tempurl%
\url{http://jemalloc.net/}
\showURL{%
\tempurl}


\bibitem[\protect\citeauthoryear{{Ferran Pallarès Roca}}{{Ferran Pallarès
  Roca}}{2017}]%
        {reductions-ferran}
\bibfield{author}{\bibinfo{person}{{Ferran Pallarès Roca}}.}
  \bibinfo{year}{2017}\natexlab{}.
\newblock \emph{\bibinfo{title}{Extending OmpSs programming model with task
  reductions: A compiler and runtime approach}}.
\newblock Bachelor's Thesis. \bibinfo{school}{Barcelona School of Informatics,
  Universitat Politècnica de Catalunya}.
\newblock


\bibitem[\protect\citeauthoryear{Fournier, Desnoyers, and Dagenais}{Fournier
  et~al\mbox{.}}{2009}]%
        {lttng}
\bibfield{author}{\bibinfo{person}{Pierre-Marc Fournier},
  \bibinfo{person}{Mathieu Desnoyers}, {and} \bibinfo{person}{Michel~R
  Dagenais}.} \bibinfo{year}{2009}\natexlab{}.
\newblock \showarticletitle{Combined tracing of the kernel and applications
  with LTTng}. In \bibinfo{booktitle}{\emph{Proceedings of the 2009 linux
  symposium}}. Citeseer, \bibinfo{pages}{87--93}.
\newblock


\bibitem[\protect\citeauthoryear{Gautier, Perez, and Richard}{Gautier
  et~al\mbox{.}}{2018}]%
        {OpenMPGranularities}
\bibfield{author}{\bibinfo{person}{Thierry Gautier}, \bibinfo{person}{Christian
  Perez}, {and} \bibinfo{person}{J{\'e}r{\^o}me Richard}.}
  \bibinfo{year}{2018}\natexlab{}.
\newblock \showarticletitle{On the Impact of OpenMP Task Granularity}. In
  \bibinfo{booktitle}{\emph{Evolving OpenMP for Evolving Architectures}},
  \bibfield{editor}{\bibinfo{person}{Bronis~R. de~Supinski},
  \bibinfo{person}{Pedro Valero-Lara}, \bibinfo{person}{Xavier Martorell},
  \bibinfo{person}{Sergi Mateo~Bellido}, {and} \bibinfo{person}{Jesus Labarta}}
  (Eds.). \bibinfo{publisher}{Springer International Publishing},
  \bibinfo{address}{Cham}, \bibinfo{pages}{205--221}.
\newblock
\showISBNx{978-3-319-98521-3}


\bibitem[\protect\citeauthoryear{Ghosh, Yan, Eachempati, and Chapman}{Ghosh
  et~al\mbox{.}}{2013}]%
        {PrototypeDependencySupport}
\bibfield{author}{\bibinfo{person}{Priyanka Ghosh}, \bibinfo{person}{Yonghong
  Yan}, \bibinfo{person}{Deepak Eachempati}, {and} \bibinfo{person}{Barbara
  Chapman}.} \bibinfo{year}{2013}\natexlab{}.
\newblock \showarticletitle{A Prototype Implementation of OpenMP Task
  Dependency Support}. In \bibinfo{booktitle}{\emph{OpenMP in the Era of Low
  Power Devices and Accelerators}},
  \bibfield{editor}{\bibinfo{person}{Alistair~P. Rendell},
  \bibinfo{person}{Barbara~M. Chapman}, {and} \bibinfo{person}{Matthias~S.
  M{\"u}ller}} (Eds.). \bibinfo{publisher}{Springer Berlin Heidelberg},
  \bibinfo{address}{Berlin, Heidelberg}, \bibinfo{pages}{128--140}.
\newblock
\showISBNx{978-3-642-40698-0}


\bibitem[\protect\citeauthoryear{{GNU Project}}{{GNU Project}}{2020}]%
        {GOMPSource}
\bibfield{author}{\bibinfo{person}{{GNU Project}}.}
  \bibinfo{year}{2020}\natexlab{}.
\newblock \bibinfo{title}{{GOMP Source Code}}.
\newblock
\newblock
\urldef\tempurl%
\url{https://github.com/gcc-mirror/gcc/tree/master/libgomp}
\showURL{%
\tempurl}
\newblock
\shownote{Accessed: 2020-02-01.}


\bibitem[\protect\citeauthoryear{Guo, Zhao, Cave, and Sarkar}{Guo
  et~al\mbox{.}}{2010}]%
        {guo2010slaw}
\bibfield{author}{\bibinfo{person}{Yi Guo}, \bibinfo{person}{Jisheng Zhao},
  \bibinfo{person}{Vincent Cave}, {and} \bibinfo{person}{Vivek Sarkar}.}
  \bibinfo{year}{2010}\natexlab{}.
\newblock \showarticletitle{SLAW: A scalable locality-aware adaptive
  work-stealing scheduler}. In \bibinfo{booktitle}{\emph{2010 IEEE
  International Symposium on Parallel \& Distributed Processing (IPDPS)}}.
  IEEE, \bibinfo{pages}{1--12}.
\newblock


\bibitem[\protect\citeauthoryear{Hendler, Incze, Shavit, and Tzafrir}{Hendler
  et~al\mbox{.}}{2010}]%
        {10.1145/1810479.1810540}
\bibfield{author}{\bibinfo{person}{Danny Hendler}, \bibinfo{person}{Itai
  Incze}, \bibinfo{person}{Nir Shavit}, {and} \bibinfo{person}{Moran Tzafrir}.}
  \bibinfo{year}{2010}\natexlab{}.
\newblock \showarticletitle{Flat Combining and the Synchronization-Parallelism
  Tradeoff}. In \bibinfo{booktitle}{\emph{Proceedings of the Twenty-Second
  Annual ACM Symposium on Parallelism in Algorithms and Architectures}} (Thira,
  Santorini, Greece) \emph{(\bibinfo{series}{SPAA ’10})}.
  \bibinfo{publisher}{Association for Computing Machinery},
  \bibinfo{address}{New York, NY, USA}, \bibinfo{pages}{355–364}.
\newblock
\showISBNx{9781450300797}
\urldef\tempurl%
\url{https://doi.org/10.1145/1810479.1810540}
\showDOI{\tempurl}


\bibitem[\protect\citeauthoryear{Karlin, Keasler, and Neely}{Karlin
  et~al\mbox{.}}{2013}]%
        {lulesh}
\bibfield{author}{\bibinfo{person}{Ian Karlin}, \bibinfo{person}{Jeff Keasler},
  {and} \bibinfo{person}{Rob Neely}.} \bibinfo{year}{2013}\natexlab{}.
\newblock \bibinfo{booktitle}{\emph{LULESH 2.0 Updates and Changes}}.
\newblock \bibinfo{type}{{T}echnical {R}eport} LLNL-TR-641973.
  \bibinfo{pages}{1--9} pages.
\newblock


\bibitem[\protect\citeauthoryear{{Klaftenegger}, {Sagonas}, and
  {Winblad}}{{Klaftenegger} et~al\mbox{.}}{2018}]%
        {queuedelegationlock}
\bibfield{author}{\bibinfo{person}{D. {Klaftenegger}}, \bibinfo{person}{K.
  {Sagonas}}, {and} \bibinfo{person}{K. {Winblad}}.}
  \bibinfo{year}{2018}\natexlab{}.
\newblock \showarticletitle{Queue Delegation Locking}.
\newblock \bibinfo{journal}{\emph{IEEE Transactions on Parallel and Distributed
  Systems}} \bibinfo{volume}{29}, \bibinfo{number}{3} (\bibinfo{year}{2018}),
  \bibinfo{pages}{687--704}.
\newblock
\urldef\tempurl%
\url{https://doi.org/10.1109/TPDS.2017.2767046}
\showDOI{\tempurl}


\bibitem[\protect\citeauthoryear{Kurzak, Ltaief, Dongarra, and Badia}{Kurzak
  et~al\mbox{.}}{2010}]%
        {taskBased1}
\bibfield{author}{\bibinfo{person}{Jakub Kurzak}, \bibinfo{person}{Hatem
  Ltaief}, \bibinfo{person}{Jack Dongarra}, {and} \bibinfo{person}{Rosa~M.
  Badia}.} \bibinfo{year}{2010}\natexlab{}.
\newblock \showarticletitle{Scheduling dense linear algebra operations on
  multicore processors}.
\newblock \bibinfo{journal}{\emph{Concurrency and Computation: Practice and
  Experience}} \bibinfo{volume}{22}, \bibinfo{number}{1}
  (\bibinfo{year}{2010}), \bibinfo{pages}{15--44}.
\newblock
\urldef\tempurl%
\url{https://doi.org/10.1002/cpe.1467}
\showDOI{\tempurl}
\showeprint{https://onlinelibrary.wiley.com/doi/pdf/10.1002/cpe.1467}


\bibitem[\protect\citeauthoryear{{LLVM Project}}{{LLVM Project}}{2020}]%
        {LLVMSource}
\bibfield{author}{\bibinfo{person}{{LLVM Project}}.}
  \bibinfo{year}{2020}\natexlab{}.
\newblock \bibinfo{title}{{LLVM OpenMP Library Source}}.
\newblock
\newblock
\urldef\tempurl%
\url{https://github.com/llvm/llvm-project/tree/master/openmp}
\showURL{%
\tempurl}


\bibitem[\protect\citeauthoryear{{Maroñas}, {Sala}, {Mateo}, {Ayguadé}, and
  {Beltran}}{{Maroñas} et~al\mbox{.}}{2019}]%
        {WorksharingMarcos}
\bibfield{author}{\bibinfo{person}{M. {Maroñas}}, \bibinfo{person}{K. {Sala}},
  \bibinfo{person}{S. {Mateo}}, \bibinfo{person}{E. {Ayguadé}}, {and}
  \bibinfo{person}{V. {Beltran}}.} \bibinfo{year}{2019}\natexlab{}.
\newblock \showarticletitle{Worksharing Tasks: An Efficient Way to Exploit
  Irregular and Fine-Grained Loop Parallelism}. In
  \bibinfo{booktitle}{\emph{2019 IEEE 26th International Conference on High
  Performance Computing, Data, and Analytics (HiPC)}}.
  \bibinfo{pages}{383--394}.
\newblock
\showISSN{1094-7256}
\urldef\tempurl%
\url{https://doi.org/10.1109/HiPC.2019.00053}
\showDOI{\tempurl}


\bibitem[\protect\citeauthoryear{Mellor-Crummey and Scott}{Mellor-Crummey and
  Scott}{1991}]%
        {10.1145/103727.103729}
\bibfield{author}{\bibinfo{person}{John~M. Mellor-Crummey} {and}
  \bibinfo{person}{Michael~L. Scott}.} \bibinfo{year}{1991}\natexlab{}.
\newblock \showarticletitle{Algorithms for Scalable Synchronization on
  Shared-Memory Multiprocessors}.
\newblock \bibinfo{journal}{\emph{ACM Trans. Comput. Syst.}}
  \bibinfo{volume}{9}, \bibinfo{number}{1} (\bibinfo{date}{Feb.}
  \bibinfo{year}{1991}), \bibinfo{pages}{21–65}.
\newblock
\showISSN{0734-2071}
\urldef\tempurl%
\url{https://doi.org/10.1145/103727.103729}
\showDOI{\tempurl}


\bibitem[\protect\citeauthoryear{Muddukrishna, Jonsson, and
  Brorsson}{Muddukrishna et~al\mbox{.}}{2015}]%
        {muddukrishna2015locality}
\bibfield{author}{\bibinfo{person}{Ananya Muddukrishna},
  \bibinfo{person}{Peter~A Jonsson}, {and} \bibinfo{person}{Mats Brorsson}.}
  \bibinfo{year}{2015}\natexlab{}.
\newblock \showarticletitle{Locality-aware task scheduling and data
  distribution for OpenMP programs on NUMA systems and manycore processors}.
\newblock \bibinfo{journal}{\emph{Scientific Programming}}
  \bibinfo{volume}{2015} (\bibinfo{year}{2015}).
\newblock


\bibitem[\protect\citeauthoryear{Olivier, Porterfield, Wheeler, Spiegel, and
  Prins}{Olivier et~al\mbox{.}}{2012}]%
        {olivier2012openmp}
\bibfield{author}{\bibinfo{person}{Stephen~L Olivier}, \bibinfo{person}{Allan~K
  Porterfield}, \bibinfo{person}{Kyle~B Wheeler}, \bibinfo{person}{Michael
  Spiegel}, {and} \bibinfo{person}{Jan~F Prins}.}
  \bibinfo{year}{2012}\natexlab{}.
\newblock \showarticletitle{OpenMP task scheduling strategies for multicore
  NUMA systems}.
\newblock \bibinfo{journal}{\emph{The International Journal of High Performance
  Computing Applications}} \bibinfo{volume}{26}, \bibinfo{number}{2}
  (\bibinfo{year}{2012}), \bibinfo{pages}{110--124}.
\newblock


\bibitem[\protect\citeauthoryear{{Perez}, {Beltran}, {Labarta}, and
  {Ayguadé}}{{Perez} et~al\mbox{.}}{2017}]%
        {JMWeaks}
\bibfield{author}{\bibinfo{person}{J.~M. {Perez}}, \bibinfo{person}{V.
  {Beltran}}, \bibinfo{person}{J. {Labarta}}, {and} \bibinfo{person}{E.
  {Ayguadé}}.} \bibinfo{year}{2017}\natexlab{}.
\newblock \showarticletitle{Improving the Integration of Task Nesting and
  Dependencies in OpenMP}. In \bibinfo{booktitle}{\emph{2017 IEEE International
  Parallel and Distributed Processing Symposium (IPDPS)}}.
  \bibinfo{pages}{809--818}.
\newblock
\urldef\tempurl%
\url{https://doi.org/10.1109/IPDPS.2017.69}
\showDOI{\tempurl}


\bibitem[\protect\citeauthoryear{Podobas, Brorsson, and Vlassov}{Podobas
  et~al\mbox{.}}{2014}]%
        {TurboBLYSK}
\bibfield{author}{\bibinfo{person}{Artur Podobas}, \bibinfo{person}{Mats
  Brorsson}, {and} \bibinfo{person}{Vladimir Vlassov}.}
  \bibinfo{year}{2014}\natexlab{}.
\newblock \showarticletitle{TurboB{\L}YSK: Scheduling for Improved Data-Driven
  Task Performance with Fast Dependency Resolution}. In
  \bibinfo{booktitle}{\emph{Using and Improving OpenMP for Devices, Tasks, and
  More}}, \bibfield{editor}{\bibinfo{person}{Luiz DeRose},
  \bibinfo{person}{Bronis~R. de~Supinski}, \bibinfo{person}{Stephen~L.
  Olivier}, \bibinfo{person}{Barbara~M. Chapman}, {and}
  \bibinfo{person}{Matthias~S. M{\"u}ller}} (Eds.).
  \bibinfo{publisher}{Springer International Publishing},
  \bibinfo{address}{Cham}, \bibinfo{pages}{45--57}.
\newblock
\showISBNx{978-3-319-11454-5}


\bibitem[\protect\citeauthoryear{Prokopec and Odersky}{Prokopec and
  Odersky}{2013}]%
        {prokopec2013near}
\bibfield{author}{\bibinfo{person}{Aleksandar Prokopec} {and}
  \bibinfo{person}{Martin Odersky}.} \bibinfo{year}{2013}\natexlab{}.
\newblock \showarticletitle{Near optimal work-stealing tree scheduler for
  highly irregular data-parallel workloads}. In
  \bibinfo{booktitle}{\emph{International Workshop on Languages and Compilers
  for Parallel Computing}}. Springer, \bibinfo{pages}{55--86}.
\newblock


\bibitem[\protect\citeauthoryear{Reed and Kanodia}{Reed and Kanodia}{1979}]%
        {10.1145/359060.359076}
\bibfield{author}{\bibinfo{person}{David~P. Reed} {and}
  \bibinfo{person}{Rajendra~K. Kanodia}.} \bibinfo{year}{1979}\natexlab{}.
\newblock \showarticletitle{Synchronization with Eventcounts and Sequencers}.
\newblock \bibinfo{journal}{\emph{Commun. ACM}} \bibinfo{volume}{22},
  \bibinfo{number}{2} (\bibinfo{date}{Feb.} \bibinfo{year}{1979}),
  \bibinfo{pages}{115–123}.
\newblock
\showISSN{0001-0782}
\urldef\tempurl%
\url{https://doi.org/10.1145/359060.359076}
\showDOI{\tempurl}


\bibitem[\protect\citeauthoryear{Roghanchi, Eriksson, and Basu}{Roghanchi
  et~al\mbox{.}}{2017}]%
        {10.1145/3132747.3132771}
\bibfield{author}{\bibinfo{person}{Sepideh Roghanchi}, \bibinfo{person}{Jakob
  Eriksson}, {and} \bibinfo{person}{Nilanjana Basu}.}
  \bibinfo{year}{2017}\natexlab{}.
\newblock \showarticletitle{Ffwd: Delegation is (Much) Faster than You Think}.
  In \bibinfo{booktitle}{\emph{Proceedings of the 26th Symposium on Operating
  Systems Principles}} (Shanghai, China) \emph{(\bibinfo{series}{SOSP ’17})}.
  \bibinfo{publisher}{Association for Computing Machinery},
  \bibinfo{address}{New York, NY, USA}, \bibinfo{pages}{342–358}.
\newblock
\showISBNx{9781450350853}
\urldef\tempurl%
\url{https://doi.org/10.1145/3132747.3132771}
\showDOI{\tempurl}


\bibitem[\protect\citeauthoryear{Sala, Rico, and Beltran}{Sala
  et~al\mbox{.}}{2020}]%
        {sala2020towards}
\bibfield{author}{\bibinfo{person}{Kevin Sala}, \bibinfo{person}{Alejandro
  Rico}, {and} \bibinfo{person}{Vicen{\c{c}} Beltran}.}
  \bibinfo{year}{2020}\natexlab{}.
\newblock \showarticletitle{Towards Data-Flow Parallelization for Adaptive Mesh
  Refinement Applications}. In \bibinfo{booktitle}{\emph{2020 IEEE
  International Conference on Cluster Computing (CLUSTER)}}. IEEE,
  \bibinfo{pages}{314--325}.
\newblock


\bibitem[\protect\citeauthoryear{Sasidharan and Snir}{Sasidharan and
  Snir}{2016}]%
        {miniAMR}
\bibfield{author}{\bibinfo{person}{Aparna Sasidharan} {and}
  \bibinfo{person}{Marc Snir}.} \bibinfo{year}{2016}\natexlab{}.
\newblock \bibinfo{booktitle}{\emph{MiniAMR - A miniapp for Adaptive Mesh
  Refinement}}.
\newblock \bibinfo{type}{{T}echnical {R}eport}.
  \bibinfo{institution}{University of Illinois}. \bibinfo{pages}{1--21} pages.
\newblock
\urldef\tempurl%
\url{http://hdl.handle.net/2142/91046}
\showURL{%
\tempurl}


\bibitem[\protect\citeauthoryear{Slaughter, Wu, Fu, Brandenburg, Garcia, Kautz,
  Marx, S.~Morris, Cao, Bosilca, Mirchandaney, Lee, Teichler, McCormick, and
  Aiken}{Slaughter et~al\mbox{.}}{2020}]%
        {taskbench}
\bibfield{author}{\bibinfo{person}{Elliott Slaughter}, \bibinfo{person}{Wei
  Wu}, \bibinfo{person}{Yuankun Fu}, \bibinfo{person}{Legends Brandenburg},
  \bibinfo{person}{Nicolai Garcia}, \bibinfo{person}{Wilhem Kautz},
  \bibinfo{person}{Emily Marx}, \bibinfo{person}{Kaleb S.~Morris},
  \bibinfo{person}{Qinglei Cao}, \bibinfo{person}{George Bosilca},
  \bibinfo{person}{Seema Mirchandaney}, \bibinfo{person}{Wonchan Lee},
  \bibinfo{person}{Sean Teichler}, \bibinfo{person}{Patrick McCormick}, {and}
  \bibinfo{person}{Alex Aiken}.} \bibinfo{year}{2020}\natexlab{}.
\newblock \showarticletitle{Task Bench: A Parameterized Benchmark for
  Evaluating Parallel Runtime Performance}. In
  \bibinfo{booktitle}{\emph{Proceedings of the International Conference for
  High Performance Computing, Networking, Storage and Analysis}} (Atlanta,
  Georgia) \emph{(\bibinfo{series}{SC ’20})}. \bibinfo{publisher}{Association
  for Computing Machinery}.
\newblock


\bibitem[\protect\citeauthoryear{{Theis} and {Wong}}{{Theis} and
  {Wong}}{2017}]%
        {TheisMoore}
\bibfield{author}{\bibinfo{person}{T.~N. {Theis}} {and}
  \bibinfo{person}{H.~.~P. {Wong}}.} \bibinfo{year}{2017}\natexlab{}.
\newblock \showarticletitle{The End of Moore's Law: A New Beginning for
  Information Technology}.
\newblock \bibinfo{journal}{\emph{Computing in Science Engineering}}
  \bibinfo{volume}{19}, \bibinfo{number}{2} (\bibinfo{year}{2017}),
  \bibinfo{pages}{41--50}.
\newblock


\bibitem[\protect\citeauthoryear{Vandierendonck, Tzenakis, and
  Nikolopoulos}{Vandierendonck et~al\mbox{.}}{2013}]%
        {DataFlowAnalysis}
\bibfield{author}{\bibinfo{person}{Hans Vandierendonck},
  \bibinfo{person}{George Tzenakis}, {and} \bibinfo{person}{Dimitrios~S.
  Nikolopoulos}.} \bibinfo{year}{2013}\natexlab{}.
\newblock \showarticletitle{Analysis of Dependence Tracking Algorithms for Task
  Dataflow Execution}.
\newblock \bibinfo{journal}{\emph{ACM Trans. Archit. Code Optim.}}
  \bibinfo{volume}{10}, \bibinfo{number}{4}, Article \bibinfo{articleno}{61}
  (\bibinfo{date}{Dec.} \bibinfo{year}{2013}), \bibinfo{numpages}{24}~pages.
\newblock
\showISSN{1544-3566}
\urldef\tempurl%
\url{https://doi.org/10.1145/2541228.2555316}
\showDOI{\tempurl}


\bibitem[\protect\citeauthoryear{Vikranth, Wankar, and Rao}{Vikranth
  et~al\mbox{.}}{2013}]%
        {vikranth2013topology}
\bibfield{author}{\bibinfo{person}{B Vikranth}, \bibinfo{person}{Rajeev
  Wankar}, {and} \bibinfo{person}{C~Raghavendra Rao}.}
  \bibinfo{year}{2013}\natexlab{}.
\newblock \showarticletitle{Topology aware task stealing for on-chip NUMA
  multi-core processors}.
\newblock \bibinfo{journal}{\emph{Procedia Computer Science}}
  \bibinfo{volume}{18} (\bibinfo{year}{2013}), \bibinfo{pages}{379--388}.
\newblock


\bibitem[\protect\citeauthoryear{Wang, Zhang, Su, Wang, Chen, Ji, and Shi}{Wang
  et~al\mbox{.}}{2014}]%
        {wang2014adaptive}
\bibfield{author}{\bibinfo{person}{Yizhuo Wang}, \bibinfo{person}{Yang Zhang},
  \bibinfo{person}{Yan Su}, \bibinfo{person}{Xiaojun Wang}, \bibinfo{person}{Xu
  Chen}, \bibinfo{person}{Weixing Ji}, {and} \bibinfo{person}{Feng Shi}.}
  \bibinfo{year}{2014}\natexlab{}.
\newblock \showarticletitle{An adaptive and hierarchical task scheduling scheme
  for multi-core clusters}.
\newblock \bibinfo{journal}{\emph{Parallel computing}} \bibinfo{volume}{40},
  \bibinfo{number}{10} (\bibinfo{year}{2014}), \bibinfo{pages}{611--627}.
\newblock


\end{thebibliography}

\appendix
\section{Artifacts Appendix}

\subsection{Getting Started}

The artifacts of this paper are provided as a docker image and can be found at the Zenodo archive: \\https://doi.org/10.5281/zenodo.4290558

To run the image, the only prerequisite is to have docker installed in your local machine.
To download and run the image, the following commands can be used:

\begin{lstlisting}
$ wget \
  https://zenodo.org/record/4290558/files/ppopp.tar
$ docker load --input ppopp.tar
$ docker run -h debian --name artifact \
  -it artifacts/ppopp:1.0.0
\end{lstlisting}

After running the earlier commands, you will be presented with an interactive prompt in an image with all the prerequisites to run the benchmarks installed.

The full suite of benchmarks can take several hours to run completely, and requires a system with a large amount of main memory (+ 32 GB), as some problem sizes are big.
To test the functionality of the artifacts, a small suite of benchmarks is provided, which will generate the granularity scaling plots with smaller problem sizes, and can be run in a few minutes.
The small suite also does only one execution of each benchmark, providing no standard deviation information.

To run the reduced set of benchmarks, the following script is provided which can be executed from \texttt{/home/user}:

\begin{lstlisting}
user@debian:~$ ./run-small-suite.sh
\end{lstlisting}

After running the benchmarks, the results corresponding to the comparison between the optimized Nanos6 runtime compared with
GNU OpenMP and LLVM OpenMP will be stored in the \texttt{/home/user/output/} folder.
To retrieve the results back to the local machine, the following command can be used \textbf{outside the container}:

\begin{lstlisting}
$ docker cp artifact:/home/user/output .
\end{lstlisting}

Which will create a folder named \texttt{output} in the current path and retrieve the plots in pdf format.

\subsection{Step by step}

The artifact is prepared to be flexible and all the sources as well as scripts to re-build all of the software are included. In this section we will explain the structure of the image, how to run the full benchmark suite, and how to change and re-build the experiments and software.

In case it is needed to install more software on the docker image, the password for the user of the machine is \texttt{user}.

\subsubsection{Full benchmark suite}
It is possible to run the full benchmark suite with the same parameters that were used on the original paper. However, be warned that the expected runtime is several hours, and that not all machines may be able to handle the input sizes due to lack of memory. To do so, run the following command:

\begin{lstlisting}
user@debian:~$ ./run-full-suite.sh
\end{lstlisting}

\subsubsection{Directory structure}

The image has the following structure on which the relevant files can be found:
\begin{itemize}
    \item \texttt{Sources}: Contains the source code for the Nanos6 runtime as well as its dependencies (Mercurium, Jemalloc and TAMPI), and GCC 9.3.0 for the benchmarks.
    \item \texttt{Benchmarks}: Contains the source code and binaries for the benchmarks, each one in a separate folder with a standalone (and working) Makefile.
    \item \texttt{Install}: Contains the built binaries for Nanos6 and its dependencies.
    \item \texttt{Automate}: Contains the Pyhton scripts and JSON configuration files that are used to execute the benchmarks.
    \item \texttt{output}: Output directory, where the plots of granularity and efficiency for each of the benchmarks are saved after running the benchmark suite.
\end{itemize}

\subsubsection{Nanos6 Sources}

Although the sources used in the article evaluation are included in the docker image, Nanos6 is free software and the most up-to-date version of the sources is publicly available on the following GitHub repository: https://github.com/bsc-pm/nanos6

We suggest using the last version from git if you want use Nanos6 in your research.

\subsubsection{Rebuilding the software}

In the root folder, an \texttt{install-all.sh} script is provided which will re-extract all the sources and re-build Nanos6, its dependencies, and the GCC 9.3.0 toolchain.
This is provided as a way to make the image re-usable, as it is possible to change the sources and re-build the whole stack.

In case you want to install Nanos6 bare-metal in a machine, we suggest to refer to the official Nanos6 documentation which can be found on the GitHub repository or inside the included source tarball, which will guide you through all the configuration and building process.

\subsubsection{The OmpSs-2 Programming Model}

The BSC Programming Models research group routinely supports other researchers that want to use or improve the OmpSs-2 programming model, the Nanos6 runtime or any of our tools. The specification for the OmpSs-2 programming model can be found online at https://pm.bsc.es/ompss-2, and you can reach us by email at \texttt{pm-tools@bsc.es}.

\subsection{Supported claims}
This artifact is designed to support the claims done in Section 6 of the paper, specifically the performance comparison between the Nanos6 runtime, containing all of the novelties presented on the paper, and other state of the art OpenMP runtimes.

The goal of the artifacts is to facilitate the reuse of our research by others, and allow access to a functional and reusable version of the sources and benchmarks referenced in the paper. Reproducing exactly the results obtained in the paper would require access to the exact same machines and software that was used, which is not in the scope of the artifact.

Running the benchmark suite will generate the performance plots based on task granularity, using the same method that was used to generate the original plots.
Raw results are also extractable, and are available on the \texttt{Automate/scaling} folder (or the \texttt{Automate/scaling\_small/} for the small suite).
However, there are some caveats and claims that are not supported by the artifacts:

\begin{itemize}
    \item The artifacts do not include non-free software that was used during the evaluation. Specifically, the following software is not included and thus not evaluated in the comparison:
    \begin{itemize}
        \item The Intel MKL library. Instead, the BLAS and LAPACK kernels used in the benchmarks have been linked against the OpenBLAS library, which may affect the results. If desired, the user can download the Intel MKL libraries from debian's non-free repositories and link the benchmarks against them.
        \item The Intel Compiler and OpenMP runtimes, which require Intel licenses to use.
        \item The AMD Optimizing C/C++ Compiler and OpenMP runtime, which is available on AMD's website.
        \item The ARM Performance Libraries, which were used on the Graviton2 machine.
    \end{itemize}
    \item Performance evaluation was done bare-metal in all the machines, without container overhead and always on exclusive nodes. Caution is advised when drawing conclusions from the results obtained running the artifacts, as similar conditions need to be achieved for the results to be valid.
\end{itemize}



\end{document}